\documentclass[runningheads]{llncs}

\usepackage{amsmath}
\usepackage{amssymb}
\usepackage{latexsym}
\usepackage{cite}

\usepackage[pagewise]{lineno}
\usepackage{graphicx}
\usepackage[loose]{subfigure}
\usepackage[ruled,vlined,linesnumbered]{algorithm2e}
\newcommand{\NP} {NP}

\begin{document}
\title{A Reduction System for Optimal 1-Planar Graphs
\thanks{Supported by the Deutsche Forschungsgemeinschaft (DFG), grant
    Br835/18-1.}
}

\author{Franz J.\ Brandenburg}
\institute{%
University of Passau,
94030 Passau, Germany \\
\email{brandenb@fim.uni-passau.de}
%brandenb@informatik.uni-passau.de
}

\maketitle

\begin{abstract}

There is a graph reduction system so that every optimal 1-planar
graph can be reduced to an irreducible extended wheel graph,
provided the reductions are applied such that the given graph class
is preserved. A graph is optimal 1-planar if it can be drawn in the
plane with at most one crossing per edge and is optimal if it has
the maximum of $4n-8$ edges.

We show that the reduction system is context-sensitive so that the
preservation of the graph class can be granted by local conditions
which can be tested in constant time. Every optimal 1-planar graph
$G$ can be reduced to every extended wheel graph whose size is in a
range from the (second) smallest one to some upper bound that
depends on $G$. There is a reduction to the smallest extended wheel
graph if $G$ is not 5-connected, but not conversely. The reduction
system has side effects and is non-deterministic and non-confluent.
Nevertheless, reductions can be computed in linear time.
\end{abstract}

%\end{frontmatter}

\section{Introduction}

Coloring graphs is a classical graph problem. The 4-color problem
for planar graphs was open for long and was   solved by Appel et
al.~ \cite{ah-epgfc-77, ahk-epgfI-77} using a compute program. A
simpler proof was given by Robertson et al.~\cite{rsst-4color-97}.
Ringel \cite{ringel-65} studied the coloring problem of 1-planar
graphs. A graph is 1-planar if it can be drawn in the plane with at
most one crossing per edge. 1-planar graphs appear when a planar
graph and its dual are drawn simultaneously. They are 6-colorable
\cite{b-np6ct-95} and the bound is tight since $K_6$ is 1-planar.

Structural properties of 1-planar graphs were first studied by
Bodendiek, Schumacher  and Wagner \cite{bsw-bs-83,bsw-1og-84,
s-s1pg-86}. They observed that  1-planar graphs with the maximum
number of edges can be obtained from 3-connected planar
quadrangulations by adding a pair of crossing edges in each
quadrangular face. Consequently, 1-planar graphs with $n$ vertices
have at most $4n-8$ edges. Bodendiek et al.~\cite{bsw-1og-84} called
1-planar graphs with $4n-8$ edges \emph{optimal} and proved that
there are such graphs for $n = 8$ and for all $n \geq 10$, and not
for $n \leq 7$ and $n = 9$.

The extended wheel graphs $XW_{2k}$ play an important role for
optimal 1-planar graphs. An \emph{extended wheel graph} $XW_{2k}$
for $k \geq 3$ consists of a circle $C = (v_1, \ldots, v_{2k})$ of
even length and two distinguished vertices $p$ and $q$, called
poles. There is an edge between   the vertices at distance two on
$C$. In addition, there is an edge between each pole and each vertex
on $C$. Note that there is no edge between the poles. Hence,
$XW_{2k}$ has $2k+2$ vertices and $8k$ edges, see
Fig.~\ref{XWgraphs}. The notation $XW_{2k}$ is taken from Suzuki
\cite{s-rm1pg-10} and is related to Schumacher's \cite{s-s1pg-86} $2
* \hat{C}_{2k}$.

\begin{figure}
   \centering
   \subfigure[The smallest extended wheel graph $XW_6$ drawn as a crossed
       cube. Any two non-adjacent
     vertices $p$ and $q$ of $XW_6$ can be taken as poles.]{
     \parbox[b]{5.5cm}{%
       \centering
       \includegraphics[scale=0.57]{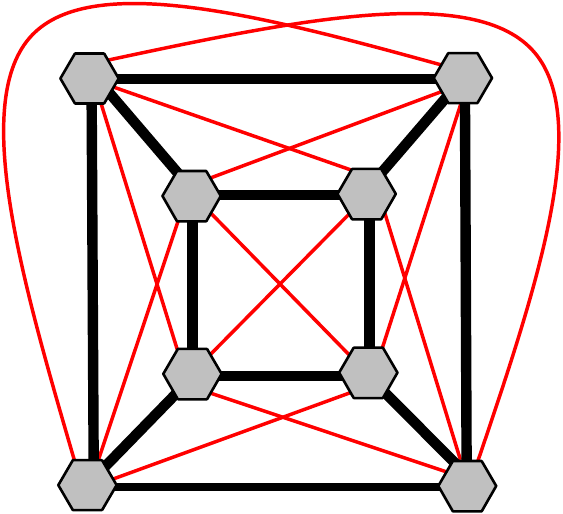}
     }
   }
   \hfil
   \subfigure[The extended wheel graph $XW_{10}$]{
     \parbox[b]{5.5cm}{%
       \centering
       \includegraphics[scale=0.34]{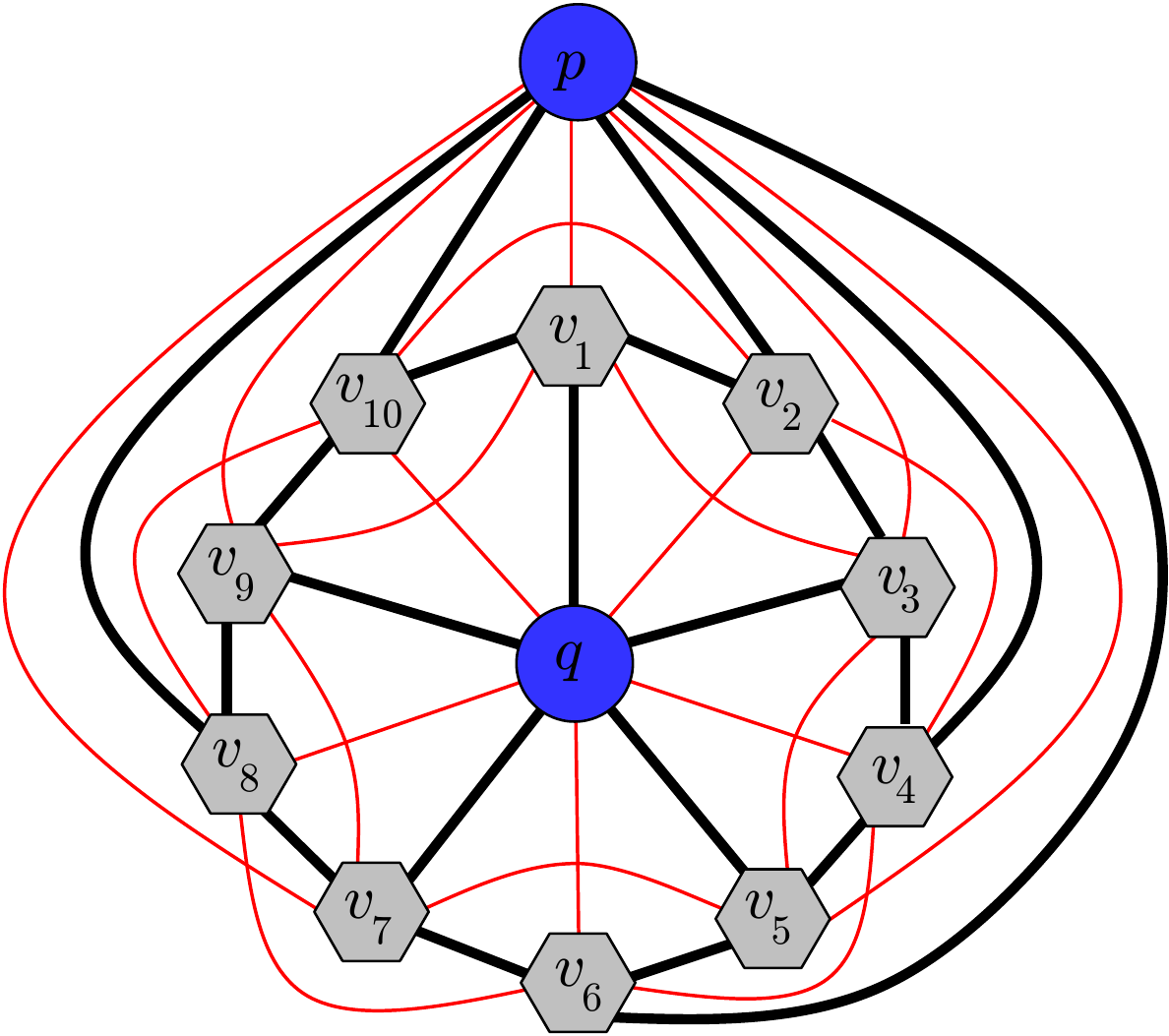}
     }
   }
   \caption{Extended wheel graphs $XW_6$ and $XW_{10}$. The vertices
   of degree six are drawn as hexagons. Planar edges are black and thick and crossed edges are red and thin.
     If the poles $p$ and $q$ change places this swaps the type of
     the incident edges between planar and crossed. The picture is unchanged, whereas there is a new embedding.}
   \label{XWgraphs}
\end{figure}

Schumacher \cite{s-s1pg-86} investigated the structure of subgraphs
of 5-connected optimal 1-planar graphs that are induced by the
vertices of degree six. He showed that these subgraphs are forests
of paths or 3-stars, except for extended wheel graphs, and provided
the following characterization: a 5-connected optimal 1-planar graph
is an extended wheel graph if and only if the subgraph induced by
vertices of degree six is a cycle. This result is no longer valid if
the precondition on 5-connected graphs is dropped. There are optimal
1-planar graph with separating 4-cycles and a cycle induced by
vertices of degree six, such as the graphs in Figs.~\ref{fig:nested}
and \ref{CR}.

\begin{figure}
   \begin{center}
   \rotatebox{270}{%
     \includegraphics[scale=0.2]{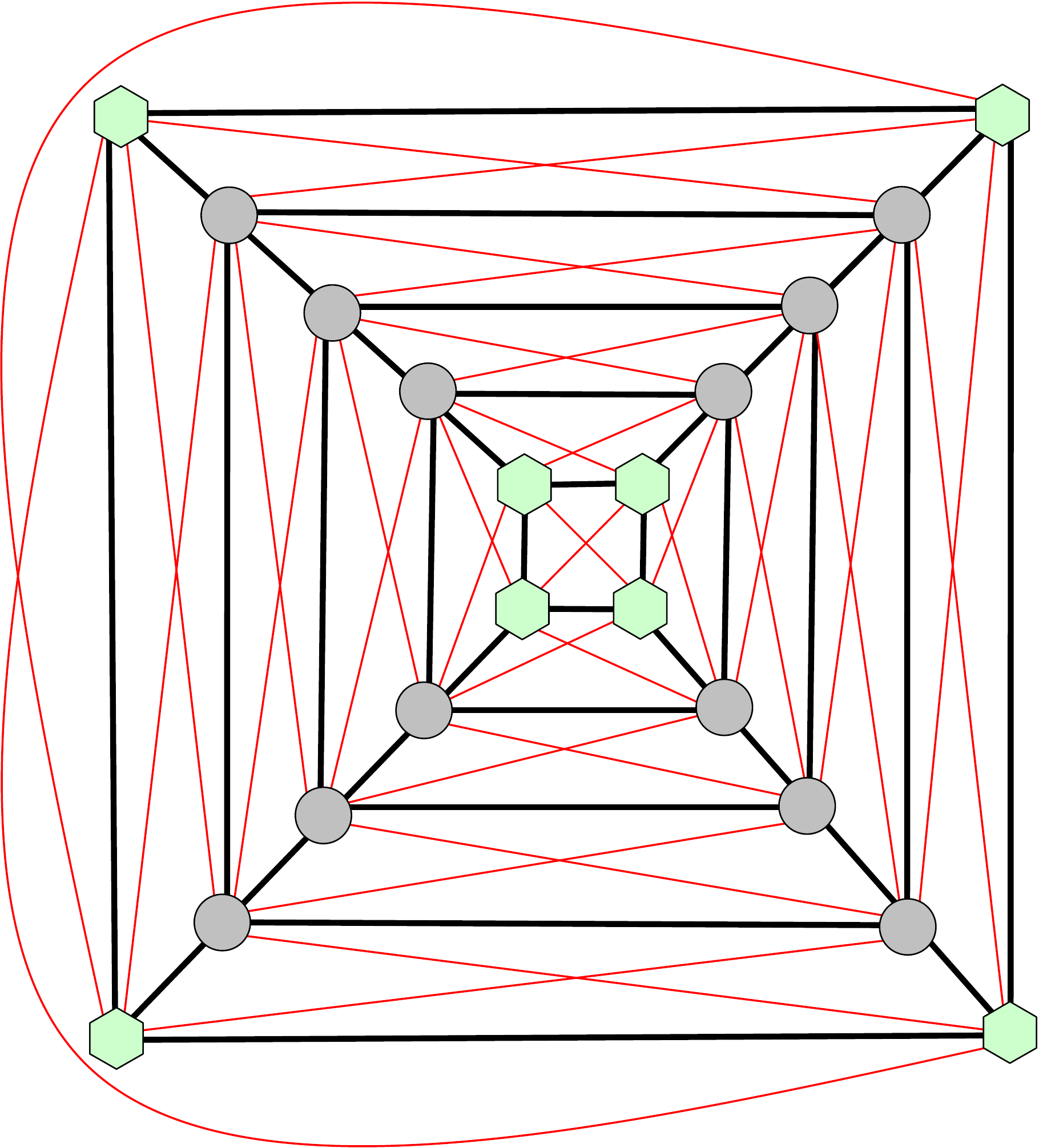}
     }
     \caption{nested crossed cubes  \label{fig:nested}}
   \end{center}
\end{figure}

Schumacher \cite{s-s1pg-86} introduced a graph reduction with a
single rule, whose  augmented version is shown in Fig.~\ref{SR} and
is called $SR$. He proved that the extended wheel graphs are
irreducible under $SR$ and that every 5-connected optimal 1-planar
graph can be reduced  to an extended wheel graph and even to $XW_8$.
However, $SR$ must be used with care. A use must preserve the given
class. This is stated  in \cite{s-s1pg-86} and  in
\cite{bggmtw-gsqs-05} for the special case of 3-connected planar
quadrangulations. It is not said, how to meet this condition. It may
need a global test for optimal 1-planarity or a 3-connected planar
quadrangulation. A false use of $SR$ to an optimal 1-planar graph
may leave the class and, in particular, may destroy the
3-connectivity of the underlying planar subgraph and thereby violate
the condition imposed by Brinkmann et al.~\cite{bggmtw-gsqs-05} on
planar quadrangulations. Moreover,
 there are uses of $SR$ that
preserve optimal 1-planarity but violate the 5-connectivity
precondition and introduce a separating 4-cycle, as Example
\ref{ex:example1} shows. Then Schumacher's reduction system gets
stuck. A subgraph with a
 4-cycle of vertices of degree six as in
Figs.~\ref{fig:nested} and \ref{CR} is inaccessible to $SR$. Hence,
Schumacher's result must be read as follows: For every 5-connected
optimal 1-planar graph $G_1$ there exist 5-connected optimal
1-planar graphs $G_1, \ldots, G_t$ for some $t \geq 1$ such that
 $G_{i+1}$ is obtained from
$G_i$ by an $SR$-reduction for $i=1,\ldots, t-1$  and $G_t$ is an
extended wheel graph.  It is left open how  $G_2, \ldots, G_t$ are
computed.

If the precondition on 5-connected graphs is dropped, then a second
reduction  is necessary, which  we call  $CR$. $CR$ is shown in
Fig.~\ref{CR} and is the inverse of the $Q_4$-cycle addition of
Suzuki \cite{s-rm1pg-10} and of the extension of the $P_3$-expansion
of Brinkmann et al.~\cite{bggmtw-gsqs-05} to 1-planar graphs. The
graph transformation rules of Schumacher and Suzuki are defined on
embeddings of 1-planar graphs and need the distinction between
planar and crossed edges. We reverse direction and consider
reductions.
%
\iffalse

We use the characterization of optimal 1-planar graphs in terms of
3-connected planar quadrangulations with pairs of crossing edges in
each face \cite{bsw-1og-84}. Brinkmann et al.~\cite{bggmtw-gsqs-05}
proved that every 3-connected planar quadrangulation can be reduced
to (or generated from) a pseudo-double wheel using two graph
reductions. A \emph{double-wheel graph}  $W_{2k}$ is the
planarization of an extended wheel graph and is a 3-connected planar
quadrangulation with edges on the circle and between the vertices at
odd (even) position on the circle and pole $p$ (pole $q$). The
reductions were generalized to optimal 1-planar graphs by Suzuki
\cite{s-rm1pg-10} and are defined below. The reductions are defined
on embeddings of 1-planar graphs. They use the distinction between
planar and crossed edges. \fi
%
The reductions are \emph{constraint} and can be used  if the given
class of graphs is preserved. Then the use is \emph{feasible}.
Fortunately, the feasibility can be checked locally and independent
of an embedding, as we shall show.  We use the terms ``good'' and
``bad'' so that a feasible reduction is applied to a good candidate.
Reductions have \emph{side effects} such that the application of a
reduction to a good candidate may change the status of other
vertices from candidate to non-candidate and of candidates from good
to bad, and vice versa. This is illustrated by shapes and colors for
the vertices in Fig.~\ref{fig:reduce1-G16}.
% and \ref{fig:reduction-XW8}.
In consequence, the reduction systems with the sets of rules $\{SR,
CR\}$ and $\{SR\}$, respectively, are constraint, context-sensitive,
 and non-confluent. These terms have been studied
in the theories of Formal Languages
 \cite{hmu-iatlc-03} and   Rewriting Systems \cite{bn-tr-98, bo-srs-93}.

In this paper, we generalize the results of Schumacher
\cite{s-s1pg-86} to arbitrary optimal 1-planar graphs and first show
how $SR$ and $CR$ can be applied feasibly. The feasibility check and
an application take only constant time. We thereby translate the
general requirement of Schumacher \cite{s-s1pg-86} and Brinkmann et
al.~\cite{bggmtw-gsqs-05} on the preservation of the given class
into an effective procedure. Then we establish that every  reducible
optimal 1-planar graph  can be reduced
 to every extended wheel $XW_{2k}$ in a range from $s$
to $t$. Here, the upper bound $t$ depends on the given graph $G$
whereas the lower bound is $s =3$ or $s =4$, where $s=3$ if $G$ is
not 5-connected. Some 5-connected optimal 1-planar graphs can also
be reduced to $XW_6$. 5-connected optimal 1-planar graphs can be
reduced using only $SR$ and there are graphs which can be reduced
using only by $CR$. The reduction system is non-deterministic and
generally admits several reductions of a reducible optimal 1-planar
graph to an irreducible extended wheel graph. Each such reduction
can be computed linear time.

 The paper is organized as follows: In the next Section we recall
some basic properties of optimal 1-planar graphs. In Section
\ref{sect:rules} we introduce  the reductions   and show how to use
them on graphs  and derive a simple quadratic-time recognition
algorithm.  Combinatorial properties of  the reduction system are
explored in Section \ref{charact}. We conclude with some open
problems on 1-planar graphs in Section \ref{sect:conclusion}.

\section{Preliminaries} \label{sect:prelim}

We consider simple  undirected graphs $G = (V,E)$ with sets of
vertices $V$ and edges $E$. The degree of a vertex is the number of
incident edges or neighbors, and the \emph{local degree} is the
number of incident edges or neighbors when restricted to a
particular induced subgraph.

A \emph{drawing} of a graph is a mapping of $G$ into the plane such
that the vertices are mapped to distinct points and the edges to
simple Jordan curves between the endpoints. It is \emph{planar} if
edges do not cross and \emph{1-planar} if each edge is crossed at
most once. A drawn graph defines an    \emph{embedding}
$\mathcal{E}(G)$ which contains all edge crossings and faces. A
drawing and an embedding are a witness for planarity and
1-planarity, respectively. For an algorithmic treatment, a planar
embedding is given by a \emph{rotation system}, which describes the
cyclic ordering of the edges incident to each vertex, or by the sets
of vertices, edges, and faces. A 1-planar embedding is given by an
embedding of the planarization of $G$, which is obtained by taking
the crossing points of edges as virtual vertices
\cite{ehklss-tm1pg-13b}.

A 1-planar embedding partitions the edges into \emph{planar} and
\emph{crossing} edges. We shall \emph{color} the planar edges black
and the crossing ones red. Other color schemes were used in
  \cite{ehklss-tm1pg-13b, el-racg1p-13, help-ft1pg-12,d-ds1pgd-13}.
  The \emph{black} or \emph{planar skeleton}
$P(\mathcal{E}(G))$ consists of the black (planar) edges and
inherits its embedding from the given 1-planar embedding. All
crossing edges are removed from $P(\mathcal{E}(G))$. A vertex $u$ is
called a \emph{black (red) neighbor} of $v$ if the edge $\{u,v\}$ is
black (red) in a 1-planar embedding. A \emph{kite} is a 1-planar
embedding of $K_4$ with a planar quadrilateral $Q$, and a pair of
crossing edges inside $Q$ and no other vertices inside $Q$. For
example, there are 3 kites in the embedding of the left graph in
Fig.~\ref{SR}. The other embedding of $K_4$ is as a tetrahedron,
whose edges, however, may be crossed \cite{b-ro1plt-16, Kyncl-09}.

Every 5-connected optimal 1-planar graph has a unique 1-planar
embedding with the exception of the extended wheel graphs, which
have two embeddings for graphs of size at least ten
\cite{s-s1pg-86}.  The different embeddings result from exchanging
the poles.
  Suzuki \cite{s-rm1pg-10} improved Schumacher's result and
dropped the 5-connectivity precondition, which is a restriction,
since optimal 1-planar graphs are 4-connected and not necessarily
5-connected. Note that a 1-planar graph is 6-connected if it is
5-connected \cite{s-s1pg-86}. All vertices of an optimal 1-planar
graph have an even degree of at least six, and there are at least
eight vertices of degree six, since in total there are $4n-8$ edges.
%The planar and the crossing edges alternate in the rotation system
%of each vertex of a 1-planar embedding.
For convenience, we shall identify a 1-planar graph and its 1-planar
embedding if the embedding is unique or clear from the context.

\section{Reductions on Graphs} \label{sect:rules}

\subsection{Previous Reductions}
 Brinkmann et al. \cite{bggmtw-gsqs-05} introduced two
graph transformations  for the generation and characterization of
3-connected planar quadrangulations. We wish to reduce graphs and
consider the inverse relations.

\begin{definition}
The $P_1$-\emph{reduction} of a 3-connected planar quadrangulation
consists of a contraction of a face $f = (u,x,v,z)$ at $x$ and $z$,
where $x$ has degree $3$ and $u,v,z$ have degree at least $3$.  The
$P_3$-\emph{reduction} removes the vertices of the inner cycle of a
planar cube, where the inner cycle is empty and vertices on the
outer cycle   have degree at least $4$.

The reductions must be applied such that they preserve the class of
 3-connected planar quadrangulations.
\end{definition}

A $P_1$-reduction is the  restriction to the black (planar) edges in
Fig.~\ref{fig:splitting}  and Fig.~\ref{SR} shows an augmented
version. The $P_3$-reduction is displayed by the black edges in
Fig.~\ref{CR}.

By the one-to-one correspondence between 3-connected planar
quadrangulations and optimal 1-planar graphs \cite{bsw-1og-84}, the
$P_1$-~and $P_3$-reductions are extended straightforwardly to
embedded 1-planar graphs, called vertex  and face contraction  by
Suzuki \cite{s-rm1pg-10}. Their inverse are called $Q_v$-splitting
and $Q_4$-cycle addition, respectively, and are used from right to
left. The illustration in Fig.~\ref{fig:splitting} is taken from
\cite{s-rm1pg-10}.

\begin{figure}
   \begin{center}
   \rotatebox{270}{%
     \includegraphics[scale=0.35]{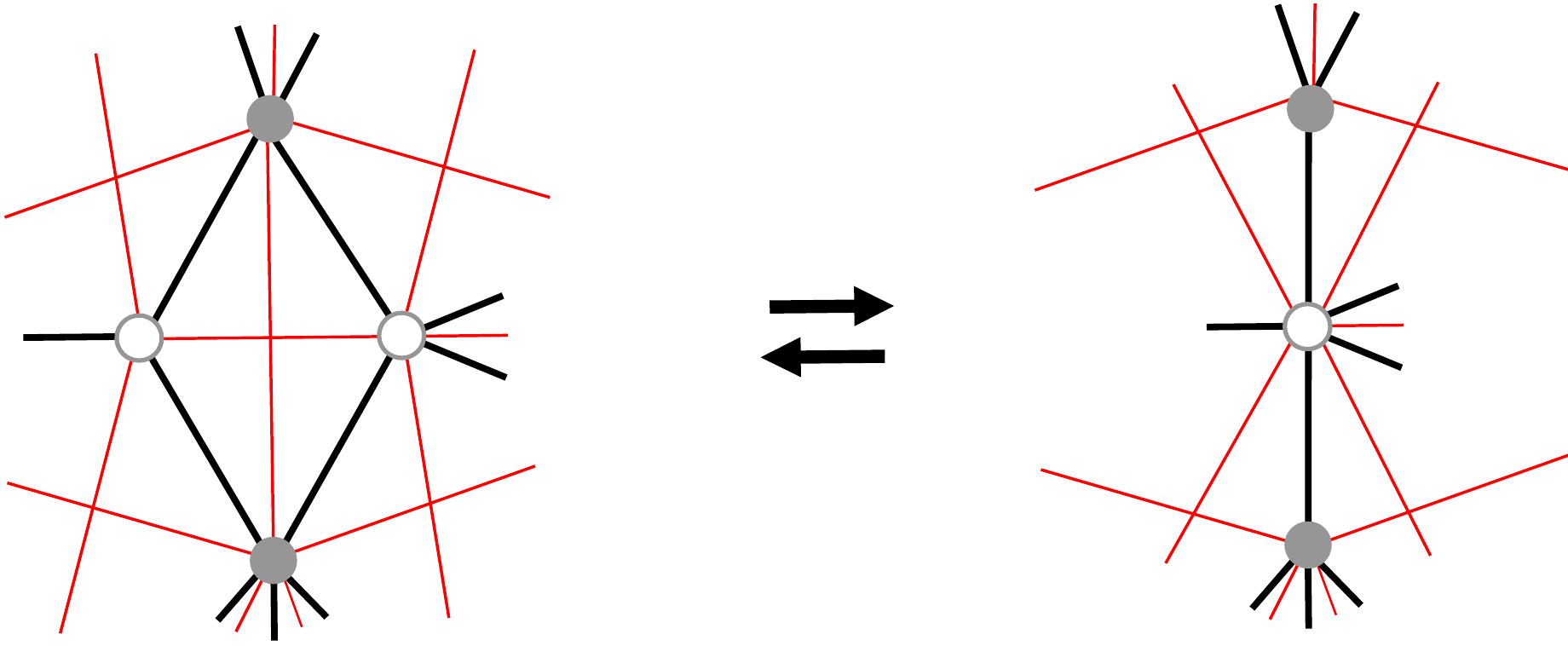}
     }
     \caption{Face contraction or  vertex splitting on 1-planar embeddings
     with planar (black and thick) and crossing (red and thin) edges
     \label{fig:splitting}}
   \end{center}
\end{figure}

By the uniqueness of embeddings of 3-connected planar
quadrangulations and of optimal 1-planar graphs (except for extended
wheel graphs) one can define the graph reductions on their
embeddings. Thereby it is assumed that the embedding is given which
is not clear for 1-planar graphs. Schumacher \cite{s-s1pg-86}
defined his ``$\hookrightarrow$''-relation on embeddings. Using the
embedding in Fig.~\ref{SR} he said that "$x$ can be merged with
$x_4$ if there is a quadrilateral of black edges and vertices $x,
x_3, x_4, x_5$ so that all paths of length four from $x$ to $x_4$
along black edges pass through $x_3$ or $x_5$". Note that each
quadrangle in an extended wheel graph has a path of length four
between opposite vertices $(v_{i-1}, v_{i+1})$ of a planar
quadrangle through one of the poles, such that the condition on
paths is violated. In consequence, the
``$\hookrightarrow$''-relation is not applicable. Suzuki
\cite{s-rm1pg-10} observed that Schumacher's \cite{s-s1pg-86}
``$\hookrightarrow$''-relation is the inverse of the
$Q_v$-splitting, since it defines the $P_1$-reduction of Brinkmann
et al.~\cite{bggmtw-gsqs-05}    on the planar  skeleton. It is said
that a reduction must preserve the given class, whereas it is not
specified  how this is achieved.

We summarize the previous results.

\begin{proposition}  \label{prop:properties}
\begin{enumerate}
  \item \cite{s-s1pg-86} \, For every 5-connected optimal 1-planar graph $G$ (with $G \neq XW_6)$
  there exists a
  reduction of $G$ to an
  wheel graph $XW_{2k}$ for some $k \geq 4$   by using $SR$ (or ``$\hookrightarrow$'') on 5-connected optimal 1-planar
  graphs.
  Every reducible 5-connected optimal 1-planar graph can be
  reduced to $XW_8$.
The extended wheel graphs are irreducible (or minimum) elements
  under $SR$.
  \item   \cite{bggmtw-gsqs-05} \,
  The class $\mathcal{Q}_3$ of all 3-connected planar quadrangulations
  is generated from the pseudo-double wheels by the
  $P_1(\mathcal{Q}_3)$ and $P_3(\mathcal{Q}_3)$-expansions. A
  pseudo-double wheel is the restriction of an extended wheel graph
  to the planar edges.
  \item \cite{s-rm1pg-10} \,
  Every optimal 1-planar graph can be obtained from an extended wheel graph
  by a sequence of $Q_v$-splittings and $Q_4$-cycle additions. The extended wheel
  graphs are irreducible under the inverse of $Q_v$-splitting and $Q_4$-cycle
  addition, i.e., under $SR$ and $CR$.
\end{enumerate}

\end{proposition}

Note that Schumacher restricted  the $\hookrightarrow$-relation to
5-connected optimal 1-planar graphs graphs. It is presupposed that
$G'$ is 5-connected if $G$ is 5-connected and  $G \hookrightarrow
G'$ holds. This presupposition is necessary, as Example
\ref{ex:example1} shows. There,  $G \hookrightarrow G'$ for optimal
1-planar graphs $G$ and $G'$, where $G$ is 5-connected and $G'$  has
a separating 4-cycle.

\iffalse

 The distinction between graphs
and embeddings is not important for the $P_1$-~and $P_3$-reductions
of Brinkmann et. al.~\cite{bggmtw-gsqs-05}, since there is a
one-to-one correspondence on $3$-connected planar graphs. They point
out that the reductions must be used with care such that the given
class of graphs is preserved. It is not specified  how this is
achieved. On the other hand, the ``$\hookrightarrow$''-relation of
Schumacher~\cite{s-s1pg-86} and the $Q_v$-splitting and $Q_4$-cycle
addition and the inverse $Q_f$-contraction and $Q_4$-removal of
Suzuki \cite{s-rm1pg-10} need a 1-planar embedding and the
distinction between planar (black) and crossing (red) edges. It is
not immediately clear how to apply these rules to graphs that are
given without an embedding. Nevertheless they characterize the
respective graphs, as stated in Propositions \ref{prop:brink} and
\ref{prop:suzuki}.

%\begin{proposition}[\!\!\cite{bggmtw-gsqs-05}]\label{prop:brink}

\begin{proposition} \cite{bggmtw-gsqs-05}
\label{prop:brink}
  The class $\mathcal{Q}_3$ of all 3-connected planar quadrangulations
  is generated from the pseudo-double wheels by the
  $P_1(\mathcal{Q}_3)$ and $P_3(\mathcal{Q}_3)$-expansions.
\end{proposition}

\begin{proposition} \cite{s-rm1pg-10}
\label{prop:suzuki}
  Every optimal 1-planar graph can be obtained from an extended wheel graph
  by a sequence of $Q_v$-splittings and $Q_4$-cycle additions. The extended wheel
  graphs are irreducible under $Q_v$-splitting and $Q_4$-cycle
  addition.
\end{proposition}

\fi

In full generality, graph reduction systems have been studied in the
theory of graph grammars \cite{c-handbookTCS-90, ce-gsmsol-12,
cmrehl-97, r-hgg-97}. A graph grammar or a graph replacement system
consists of a finite set of graph transformations or   rules. Each
rule  is a pair of graphs $(L, R)$ and $L$ is said to be replaced by
$R$. A graph $L$ occurs in $G$ and $L$ is said to \emph{match} a
subgraph $H$ of $G$ if there is a graph homomorphism between $L$ and
$H$, which is one-to-one on the vertices and edges but not
necessarily onto for the edges. Sometimes this is even relaxed. If
$x$ and $y$ are two vertices of $L$ with $x \neq y$ and $x'$ and
$y'$ are their matched counterparts in $H$ then $x' \neq y'$, and
there is an edge $\{x',y'\}$ if there is an edge $\{x,y\}$. An
application of $(L,R)$ to a graph $G$ replaces an occurrence of $L$
by an occurrence of $R$ while the remainder $G-L$ is preserved. It
results in a graph $G'$ which contains a subgraph $H'$ matching $R$.
There are edges between vertices of $G-L$ and $R$ according to some
conditions that are given with $(L,R)$. In our case, the graphs $L$
and $R$ have a common outer frame and the edges between the frame
and $G-L$ are kept.

Graph grammars generally operate on labeled graphs where the
vertices and edges are labeled by symbols from an alphabet. This
resembles context-free grammars on strings. The labels are used to
distinguish vertices and edges of a graph and to regulate the
application of a graph transformation. A graph grammar is used to
generate a graph language, which is a set of labeled graphs. Typical
(unlabeled) graph languages are the sets of binary trees,
series-parallel graphs, or complete (bipartite) graphs.

\iffalse

 However, there is no nice graph grammar that generates
the set of planar graphs. Nice means that there are only a few and
well-understandable transformations. In general, the power of Turing
machines can be encoded into the graph transformations of graph
grammars. If nice means context-free \cite{e-cfgg-97}, then there is
no context-free graph grammar whose language includes infinitely
many square grids \cite{b-ptgg-87}. The result also holds for planar
quadrangulations. Hence, our graph rewriting systems with rules
$P_1$ and $P_3$ (or $Q_v$-splitting and $Q_4$-cycle addition or $SR$
and $CR$) are not context-free in this setting.

%  \section{Reductions on Graphs} \label{sect:rules}

Suzuki \cite{s-rm1pg-10} observed that Schumacher's
``$\hookrightarrow$''-relation \cite{s-s1pg-86} is the inverse of
the  $Q_v$-splitting, since it defines the $P_1$-reduction of
Brinkmann et al.~\cite{bggmtw-gsqs-05}    on the planar  skeleton.
\fi

\subsection{Our Reductions}

For our study we augment the relations of Schumacher and Suzuki to
small graphs from which a part is removed and a part is kept as
context. Thereby, the feasibility of their use on graphs can be
expressed on the augmented graph. We make use of the uniqueness of
1-planar embeddings of optimal 1-planar graphs that are not extended
wheel graphs and call such graphs \emph{reducible}.

We reverse the expansions of Brinkmann et al. and Suzuki and call
the augmented version \emph{SR-reduction} (Schumacher reduction) and
\emph{CR-reduction} (crossed cube reduction), or just $SR$ and $CR$,
and the graphs of the left $CS$ (crossed star) and $CC$ (crossed
cube), respectively. The $SR$-reduction augments the vertex
splitting of Suzuki and includes the subgraph induced by the center
$x$. The reductions are shown in Figs.~\ref{SR} and \ref{CR}
including a 1-planar embedding and an edge coloring. The tiny
strokes at the outer vertices indicate further edges  which are
necessary. These vertices may have even more edges to outer
vertices. The reductions should be applied to an optimal 1-planar
graph.
 The reverse transformations (from right to left) are $SR^{-1}$
and $CR^{-1}$, respectively.

\begin{definition}
A $CR$-reduction consists of a crossed cube as left-hand side $L$
and replaces it by a kite as right-hand side $R$, as shown in
Fig.~\ref{CR}. Simply speaking the inner cycle is removed, and the
outer 4-cycle is the context. Similarly, an $SR$-reduction consists
of a graph $CS$ with three kites that meet at a vertex $x$ of degree
six as $L$ and it removes $x$ and replaces $CS$ by two adjacent
kites, as shown in Fig.~\ref{SR}. The neighbors of $x$ are the
context.

A graph $G$ reduces to a graph $G'$ if there is an induced subgraph
$H$ of $G$ that matches $L$ and an induced subgraph $H'$ of $G'$
that matches $R$ so that $G' = G-L+R$. The context of $L$ and $R$ is
kept.
 Let $G \rightarrow G'$ if $G$
reduces to $G'$ by    $SR$- or   $CR$, and let ``$\rightarrow^*$
denote the reflexive and transitive closure.
\end{definition}

Schumacher \cite{s-s1pg-86} and  Brinkmann et al.
\cite{bggmtw-gsqs-05}   require that the use of a reduction
preserves the given class. Thereby an application of a reduction is
\emph{constrained}. An infeasible application destroys the
3-connectivity of the underlying planar skeleton of 1-planar
embeddings or introduces multiple edges, which ultimately leads to a
violation of 3-connectivity.

As our first result, we show that the feasibility  of  an $SR$- or
$CR$-reduction can be expressed by local properties of the matched
left hand side graphs and can be checked in constant time.

\begin{definition}
  A vertex $x$  of an optimal 1-planar graph $G$  of degree six is called a
  \emph{candidate}.
A candidate $x$ is \emph{``good''} if there is a feasible
application of a reduction at $x$.

  In case of $SR$, $x$ is
  the center of a subgraph $H(x)$ of $G$ that matches $CS$, and there is
  a  red neighbor $v$, called  a \emph{target},  so that $SR$ can be applied by merging $x$
  with $v$, denoted $SR(x \mapsto v)$.
The application is \emph{feasible} and $SR(x \mapsto v)$ is
  \emph{good} if $G$ does not contain any of the
  edges $\{x_2, x_4\}, \{x_6,x_4\}$ and $\{x_1,
x_4\}$, where the vertices of $H(x)$ and $CR$ are identified and
$v=x_4$. Then $x$ is drawn as a green hexagon and the reduction is
indicated by an arrow.

An edge $\{x_2, x_4\}$ or $ \{x_6,x_4\}$ in $H(x)$ is called a
\emph{blocking red edge} and $ \{x_1,x_4\}$   is called a
\emph{blocking black edge}. Note that a candidate $x$ has three red
neighbors for a feasible $SR$ reduction, namely $x_2, x_4$ and
$x_6$.

  In case of   $CR$   there are four candidates $x_1, x_2, x_3, x_4$
that are matched by the vertices of the inner cycle of $CC$, and
these vertices and their neighbors match $CC$. We denote the use by
$CR(x_1, x_2, x_3, x_4)$. An application is \emph{feasible} if the
edges $\{v_1, v_3\}$ and $\{v_2,  v_4\}$ are missing in the matched
subgraph of $CC$. Then the vertices matched  by $x_1, x_2, x_3, x_4$
are drawn as blue hexagons. If the subgraph matched by $CC$ also
contains an edge $\{v_1, v_3\}$ or $\{v_2,  v_4\}$, then these edges
are \emph{blocking red edges} and $CR(x_1, x_2, x_3, x_4)$ is
infeasible.

Otherwise, a candidate is   \emph{``bad''}   and is drawn orange.
Vertices of degree at least eight are non-candidates  and are drawn
as dark circles.
 \end{definition}

\iffalse

Then the reductions $SR(x \mapsto v)$ are bad for all three red
neighbors   of $x$.
  A bad reduction $SR(x \mapsto v)$ is
  \emph{blocked} by a vertex $u$ if $u$ is a black neighbor of $x$ and $v$
  of degree six and if $u$ is any vertex on the outer circle of degree six in case of
  $CR$, respectively.
%
  An edge $e=(u, v)$ of $H(x)$ is a \emph{blocking red edge} of
  $SR(x \mapsto v)$ if $u$ is a red neighbors of $v$. Edge $e$ is a
  \emph{blocking black edge} if $u$ is a black neighbor and
  $e$ is not matched by an edge of $CS$.
If the outer cycle  of $CC$ matches $(v_1,v_2, v_3, v_4)$, then
edges $(v_1, v_3)$ and $(v_2, v_4)$ are \emph{blocking red edges} of
a $CR$-reduction.

\fi

\begin{figure}
   \begin{center}
     \includegraphics[scale=0.40]{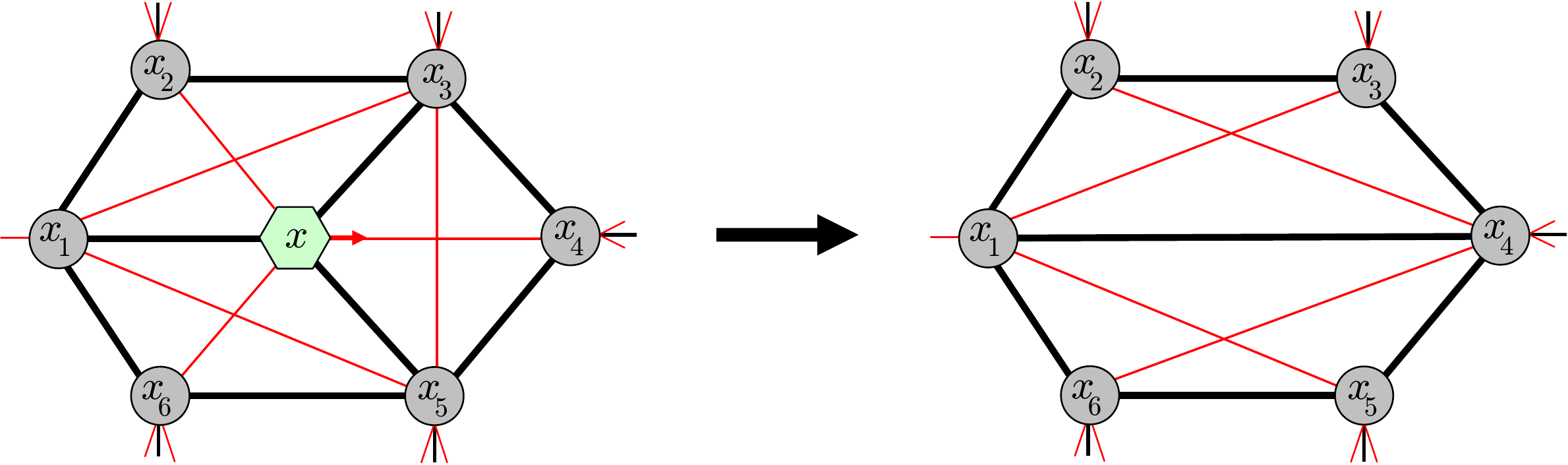}
     \caption{A $SR$-reduction for optimal 1-planar graphs.
     The candidate is drawn as a hexagon and other vertices as circles.
     Planar edges are drawn black and thick and crossing edges red and thin. A good candidate is colored
      green. There are three possible $SR$-reductions at $x$ towards
     its red neighbors,
     $SR(x  \mapsto x_2), SR(x  \mapsto x_6)$ and $SR(x  \mapsto x_4)$. Here, the latter is applied
     which is indicated by the arrow at $x$.
     The
       tiny strokes at the outside indicate further necessary edges. The left
       graph is $CS$ together with its 1-planar embedding.
       \label{SR}}
   \end{center}
\end{figure}

\begin{figure}
   \begin{center}
     \includegraphics[scale=0.33]{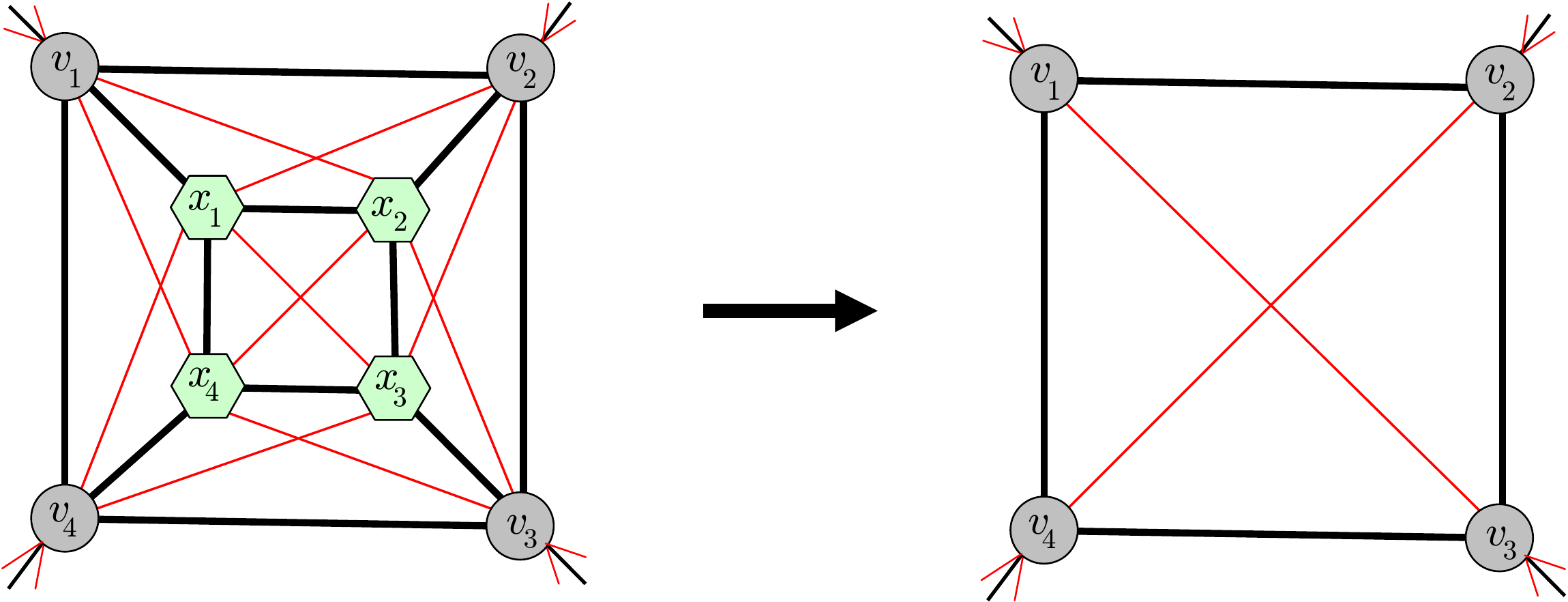}
     \caption{The reduction $CR(x_1, x_2, x_3, x_4)$ for optimal 1-planar
       graphs with candidates $x_1, x_2, x_3, x_4$. The reduction is good and the candidates are colored light
       green
       if there is no   edge $(v_1, v_3)$ or $(v_2, v_4)$ in the outer face. The
       graph on the left is $CC$ together with its embedding.\label{CR}}
   \end{center}
\end{figure}

An infeasible  application of a reduction would introduce a multiple
edge, which is illegal, since the graphs are simple. If there is a
blocking black edge $\{x_1, x_4\}$ and $SR(x \mapsto x_4)$ were
applied, then vertices $x_1$ and $x_4$ are a separation pair for the
planar skeleton after the reduction, which violates the
3-connectivity of the planar skeleton   as imposed by Brinkmann et
al.~\cite{bggmtw-gsqs-05}. Similarly, a blocking red edge $\{x_2,
x_4\}$  and a use of $SR(x \mapsto x_4)$ implies that vertices $x_2$
and $x_4$  are a separation pair for the planar skeleton. The red
blocking edge $\{x_2, x_4\}$ is crossed by some red edge $\{z_1,
z_2\}$ such that there is a planar quadrangle with the vertices
$x_2, x_3, x_4, z_1$ and $x_3$ is connected with $z_1$ inside the
quadrangle. If $SR(x \mapsto x_4)$ were applied, then all planar
paths from $x_3$ or $z_1$ to $x_5$ must pass $x_2$ or $x_4$.
Forthcoming, we assume that the reduction is feasible if $G
\rightarrow G'$ holds.

Surprisingly, an  infeasibility is directly  related to blocking
edges and  can be recognized by  an inspection of the subgraphs
matched by $CS$ and $CC$, respectively.

\begin{lemma} \label{lem:feasibility}
If $G$ is an optimal 1-planar graph and $G \rightarrow G'$ is
 feasible, then $G'$ is an optimal 1-planar graph.
\end{lemma}

%\noindent \textbf{Proof.}
\begin{proof}
We use the fact that optimal 1-planar graphs have a unique embedding
(except for extended wheel graphs) \cite{s-s1pg-86, s-rm1pg-10} and
identify a graph and its 1-planar embedding as well as the graphs of
$CS$ and $CC$ and their matched subgraphs in $G$ and $G'$. Clearly,
the embedding
  of $G'$ is 1-planar if the embedding of $G$ is 1-planar. Also,
$G'$ has $4n'-8$ edges, where $n'$ is the size of $G'$.

It remains to show that the planar skeleton
$\mathcal{P}(\mathcal{E}(G'))$ remains 3-connected.

Suppose that  $SR(x \mapsto x_4)$ is feasibly applied to a candidate
$x$ of $G$ with neighbors $x_1,\ldots, x_6$ in circular order.
First,   there are three vertex disjoint paths from $x_3$ to $x_5$
in $G'$ using only planar (black) edges, $p_1 = (x_3, x_4, x_5)$,
$p_2 = (x_3, x_2, x_1, x_6, x_5)$ and a third path $p_3 = (x_3, z_1,
z_2, \ldots, z_{2r}, z_{2r+1}, x_5)$ for some $r \geq 1$, where $(x,
x_3, z_1, \ldots, z_{2r+1}, x_5)$ are the neighbors of $x_4$ in $G$
in circular order. Since the edges $\{x_1,x_4\}, \{x_2, x_4\}$ and
$\{x_6, x_4\}$ are missing, the vertices $z_1, \ldots, z_{k+1}$ are
distinct from $x_1, x_2$ and $x_6$, and therefore path $p_3$ is
vertex disjoint to $p_1$ and $p_2$. In retrospect, there are four
vertex disjoint paths between $x_3$ and $x_5$ in $G$ using only
planar (black) edges.

Next, consider two vertices $u$ and $v$ with $\{u, v\} \neq \{x_3,
x_5\}$ in $G$ and $G'$. By assumption, there are three vertex
disjoint paths between $u$ and $v$ in $G$ through planar edges.
Suppose that one of them, say $p_2$, passes through vertex $x$. Then
$p_2$ passes through $x_3$ or $x_5$ or both. Otherwise, the same
paths can be taken in $G'$.
Suppose that $p_2$ passes through $x_1, x$ and $x_3$. The other
cases are similar. If $x_2$ is not passed by $p_1$ and $p_3$, then
reroute $p_2$ through $x_2$. Otherwise, reroute $p_2$ through $x_4$,
both in $G$ and $G'$. If $x_4$ is on path $p_3$ (or $p_1$) it uses
the edge $\{x_4, x_5\}$ and is rerouted as follows. Path $p_3$ and
the third path between $x_3$ and $x_5$ meet at some vertex $z_i$.
Then the segment $z_i, \ldots, x_4, x_5$ of $p_3$ is replaced by a
detour through $z_i, \ldots, z_{2r+1}, x_5$. Hence, there are three
vertex disjoint
 paths between $u$ and $v$ in $G'$ using only planar edges,
 and $\mathcal{P}(\mathcal{E}(G'))$
is 3-connected.

Similarly, if $CR(x_1,x_2, x_3, x_4)$ is feasible and $v_1, \ldots,
v_4$ are the vertices of the outer cycle, then there are four vertex
disjoint path through planar edges between   $v_i$ and $v_j$ in $G$,
namely two along the outer cycle, one through vertices of the inner
cycle and a forth outside the subgraph matched by $CC$.  If $v_i$
and $v_j$ are adjacent on the outer cycle, then there is an outer
quadrangle $(v_i, v_j, u, v)$ since there is no blocking red edge.
Otherwise, if $v_i$ and $v_j$ are antipodal, then there is a path
through planar edges $v_i, z_1, \ldots, z_{2r+1}, v_j$ where $v_i,
z_1, \ldots, z_{2r+1}, v_j$ are the neighbors of $v_k$ in circular
order excluding the vertices of the inner cycle and $v_k$ is between
$v_i$ and $v_j$ on the outer cycle, since there is no edge $(v_k,
v_{k+2}$. In consequence, there (at least) three vertex disjoint
paths through planar edges between $v_i$ and $v_j$ in  $G'$.
%
%\qed
%\hspace{5mm} $\square$ \\
\end{proof}

 We have shown \cite{b-ro1plt-16} that a feasible reduction uniquely determines the
embedding of the matched graph of $CS$ and $CC$ and the feasibility
is obtained from the degree vector.

\begin{definition}
  Let $x$ be a candidate of a graph $G$ and let $H(x)$ be
   the subgraph induced by $x$ and its six neighbors.
   Let $\overrightarrow{H(x)} = (d_1, \ldots, d_7)$ be
  the lexicographically ordered 7-tuple of local degrees restricted to
  $H(x)$, and call $\tau(x) = d_1$   the
  \emph{type} of $x$.
\end{definition}

For example, candidate $x$ of $CS$ in Fig.~\ref{SR} has
$\overrightarrow{H(x)} = (3,3,3,5,5,5,6)$ if there are no additional
edges $\{x_i, x_j\}$ for $i \neq j$ in the outer face.
 A candidate $x$ of the inner cycle of $CC$ and of the cycle of an
extended wheel graph $XW_{2k}$ for $k \geq 4$ has
$\overrightarrow{H(x)} = (4,4,5,5,5,5,6)$.

\begin{lemma}   \label{apply-SR-CR}  \cite{b-ro1plt-16}
  A candidate $x$ of an optimal 1-planar graph is good for $SR$ if and only
  if $\tau(x) = 3$ and $SR(x \mapsto v)$ is good for every vertex $v$ of
  local degree three.

  A candidate $x_1$ of an optimal 1-planar graph is good for $CR$ if and only
  if $\tau(x_1) = 4$ and $x_1$ has three neighbors $x_2, x_3, x_4$
  which are candidates and
 $\overrightarrow{H(x_i)} = (4,4,5,5,5,5,6)$ for $i=1,2,3,4$.
\end{lemma}

The existence of a good candidate is granted unless all candidates
are bad, as in an extended wheel graph, or if the graph is not
optimal 1-planar.

\begin{lemma} \label{lem:existence} If $G$ is an optimal 1-planar graph, which is
  not an extended wheel graph,
  and $C$ is a separating $4$-cycle which partitions $G-C$ into
  $G_{in}$ and $G_{out}$,
  then there is a    good candidate  in $G_{in}$ (and   in $G_{out}$).
\end{lemma}

%\noindent \textbf{Proof.}
\begin{proof}
  According to Brinkmann et al. \cite{bggmtw-gsqs-05} there is a good
  candidate $x$ for their $P_1$-~and $P_3$-reductions (or expansions) on
  3-connected planar quadrangulations unless the graph is a double-wheel graph.
  In Lemma 4 they prove that a good candidate lies in  the innermost
  (or outermost) separating $4$-cycle. By the one-to-one correspondence
  between  3-connected planar quadrangulations and optimal 1-planar graphs,
  this transfers to optimal 1-planar graphs.
  %\hspace{5mm} $\square$ \\
 %\qed
\end{proof}

\iffalse

\begin{lemma} \label{XW-test} There is a linear time algorithm to test
  whether a graph is an extended wheel graph.
\end{lemma}

\begin{proof}
  If the input graph $G$ has eight vertices, we check that $G = XW_{6}$ by
  inspection. For $k \geq 4$, an extended wheel graph $XW_{2k}$ has two
  poles $p$ and $q$ of degree $2k$ as distinguished vertices and a cycle
  of $2k$ vertices of degree six. This is checked in a
  preprocessing step on the given graph and takes en passant $O(1)$ time.
  For a final check, we remove the poles and restrict ourselves to the subgraph $H$ induced
  by the vertices of degree six. The cyclic ordering
  of the vertices of $H$ and an embedding are determined as follows:
   Each  vertex $v$ has four neighbors in $H$ and
  the cyclic ordering of these vertices can be determined as $(v_{-2}, v_{-1},
  v, v_{+1}, v_{+2})$ by the missing edges $(v_{-2}, v_{+1}), (v_{-2},
  v_{+2})$ and $(v_{-1}, v_{+2})$.  Altogether, the tests take $O(2k)$ time.
%  \qed
\end{proof}

\fi

From Lemmas \ref{lem:feasibility}-\ref{lem:existence} we obtain a
simple quadratic-time recognition algorithm for optimal 1-planar
graphs, which searches the graph for a good candidate, checks the
feasibility of a reduction,   applies it, and thereby removes one or
four vertices. Finally, it checks whether the obtained graph is an
extended wheel graph. The search for a good candidate has been
improved by some book-keeping technique such that there is an
asymptotically optimal algorithm for the recognition of optimal
1-planar graphs.

\begin{proposition} \label{thm:quadratic} \cite{b-ro1plt-16}
There is a linear-time recognition algorithm for optimal 1-planar
graphs.
\end{proposition}

\begin{example} \label{ex:example1}
For an explanation of the reductions consider the input graph
$G_{16}$ as shown in Fig.~\ref{fig:G16} with a 1-planar embedding.
The graph does not have a separating 4-cycle and thus is 5-connected
and even 6-connected \cite{s-s1pg-86}. Initially, vertices $d, e, h,
i, j, n, o$ are good candidates for an $SR$-reduction, and the good
reductions are indicated by a red arrow. In this example, a
candidate is good for an $SR$-reduction if it has two black
neighbors of degree at least eight. Vertices $c$ and $m$ are bad
candidates, since at least two black neighbors have degree six and
there are red blocking edges for $SR(x \mapsto v)$ and each red
neighbor.

If we first apply $SR(h \mapsto i)$, then there is a separating
4-cycle $(a,b,i,g)$ and $c,d,e,f$ are good candidates for a
$CR$-reduction, as shown in Fig.~\ref{fig:G16hi}. Vertex $l$ becomes
a new good candidate and vertex $i$ has degree 8 and is no longer a
candidate. Figs.~\ref{fig:G16hi-nl} and ~\ref{fig:G16hi-nl-kj} show
the reductions  $SR(n \mapsto l)$ and $SR(k \mapsto j)$. Then $g$
has only degree six and $\{a,i\}$  is a red blocking for $CR$.
Therefore, vertices   $c,d,e,f$ change from good (green) to bad
(orange).  The reduction $SR(j \mapsto g)$ in
Fig.~\ref{fig:G16hi-nl-kj-jg} is like an  undo for $c,d,e,f$ , which
are removed by $CR(c,d,e,f)$ and result in the minimum extended
wheel graph $XW_6$.

An alternative reduction uses   $SR(d \mapsto i), SR(n \mapsto l),
SR(h \mapsto k), SR(g \mapsto j), SR(k \mapsto m)$, and finally
$SR(l \mapsto f)$ (or $SR(j \mapsto o)$. Then all intermediate
graphs  are 5-connected and $XW_8$ is the final result.
\end{example}

\begin{figure}
  \centering
  \subfigure[A 1-planar embedding of $G_{16}$. Good candidates are drawn as
      a hexagon  and a good  $SR$-reduction $SR(x \mapsto v)$
      is indicated by an arrow. Bad candidates are colored orange.
      Non-candidates of degree at least $8$ are drawn as
      circles.]{
     \parbox[b]{\textwidth}{%
       \centering
      \includegraphics[scale=0.33]{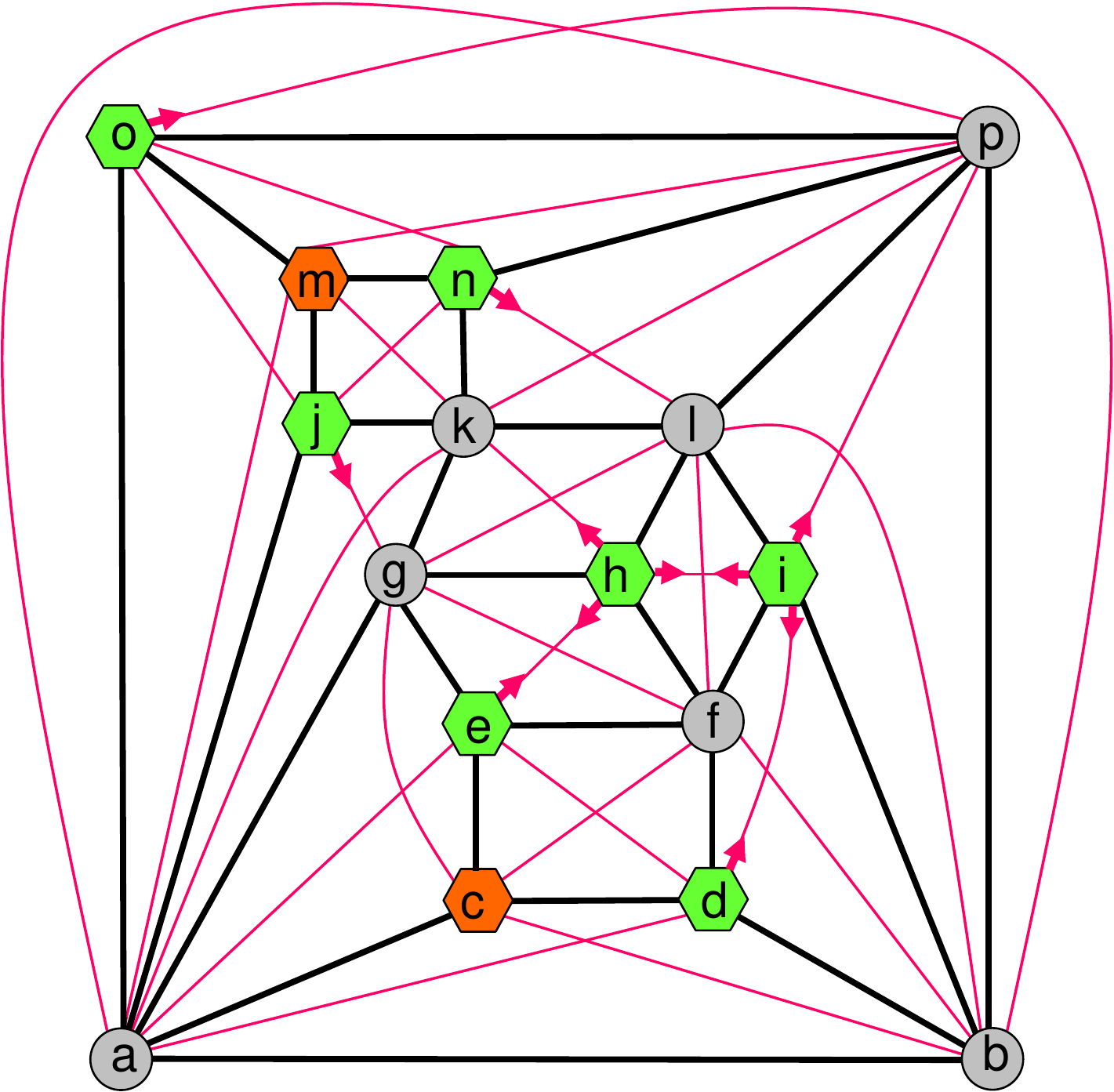}
    }
    \label{fig:G16}
  }
\end{figure}

\begin{figure}
  \centering
  \subfigure[Graph $G_{16}$ after $SR(h \mapsto i)$]{
    \parbox[b]{5.5cm}{%
      \centering
      \includegraphics[scale=0.28]{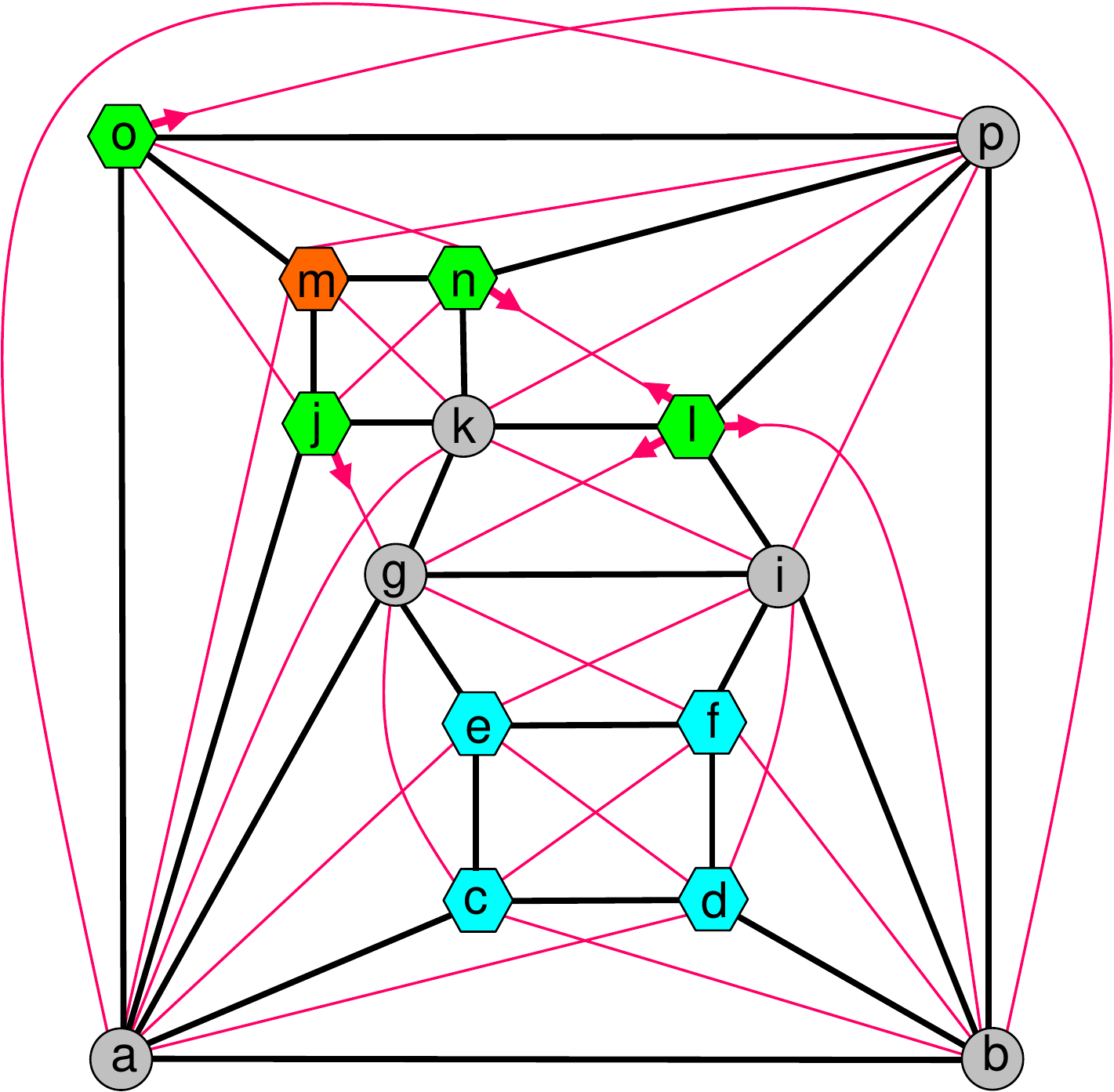}
    }
    \label{fig:G16hi}
  }
  \hfil
  \subfigure[and  after $SR(n \mapsto l)$]{
    \parbox[b]{5.0cm}{%
      \centering
      \includegraphics[scale=0.28]{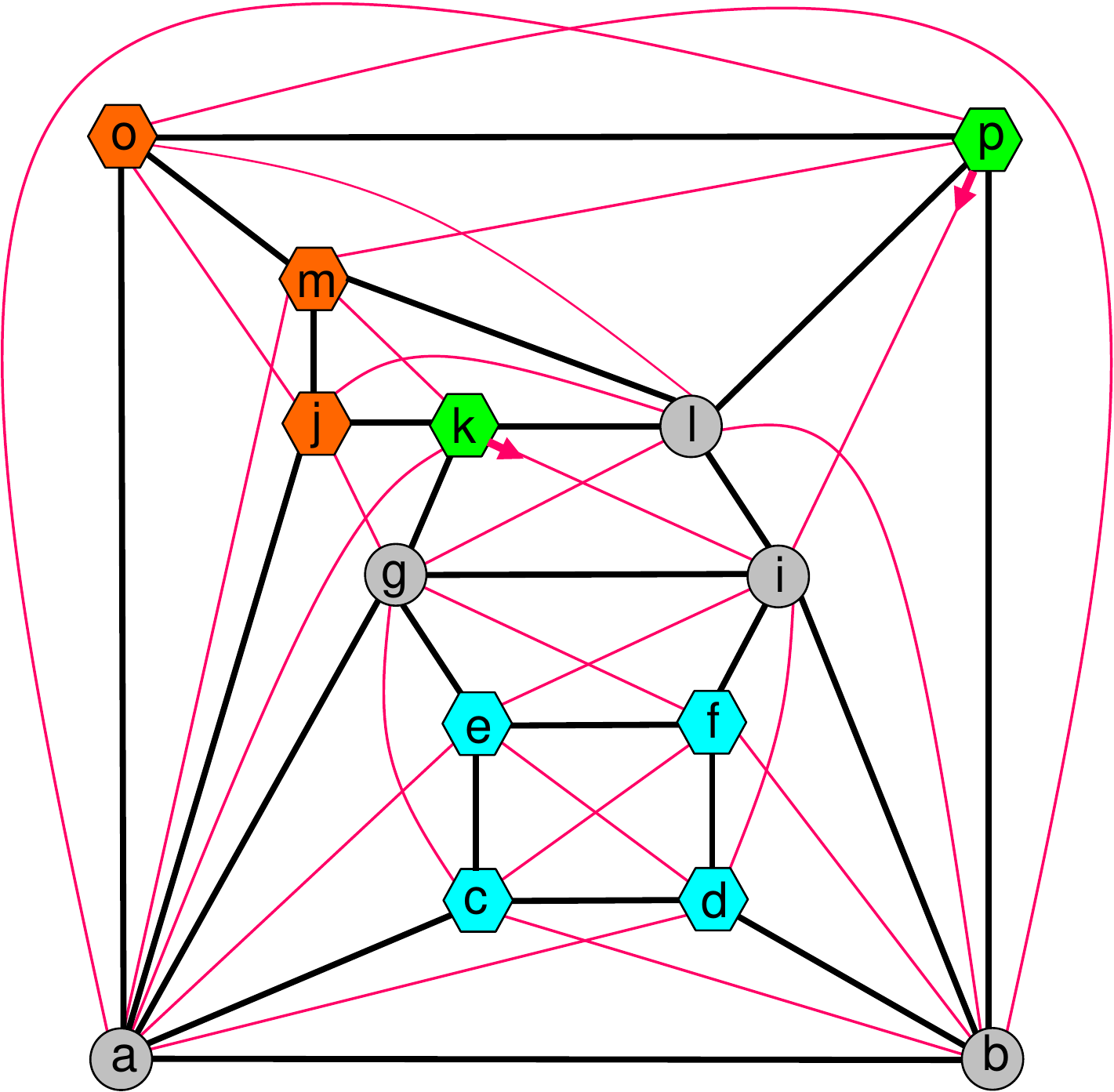}
    }
    \label{fig:G16hi-nl}
  }
  \subfigure[and  after $SR(k \mapsto i)$. ]{
    \parbox[b]{5.5cm}{%
      \centering
      \includegraphics[scale=0.28]{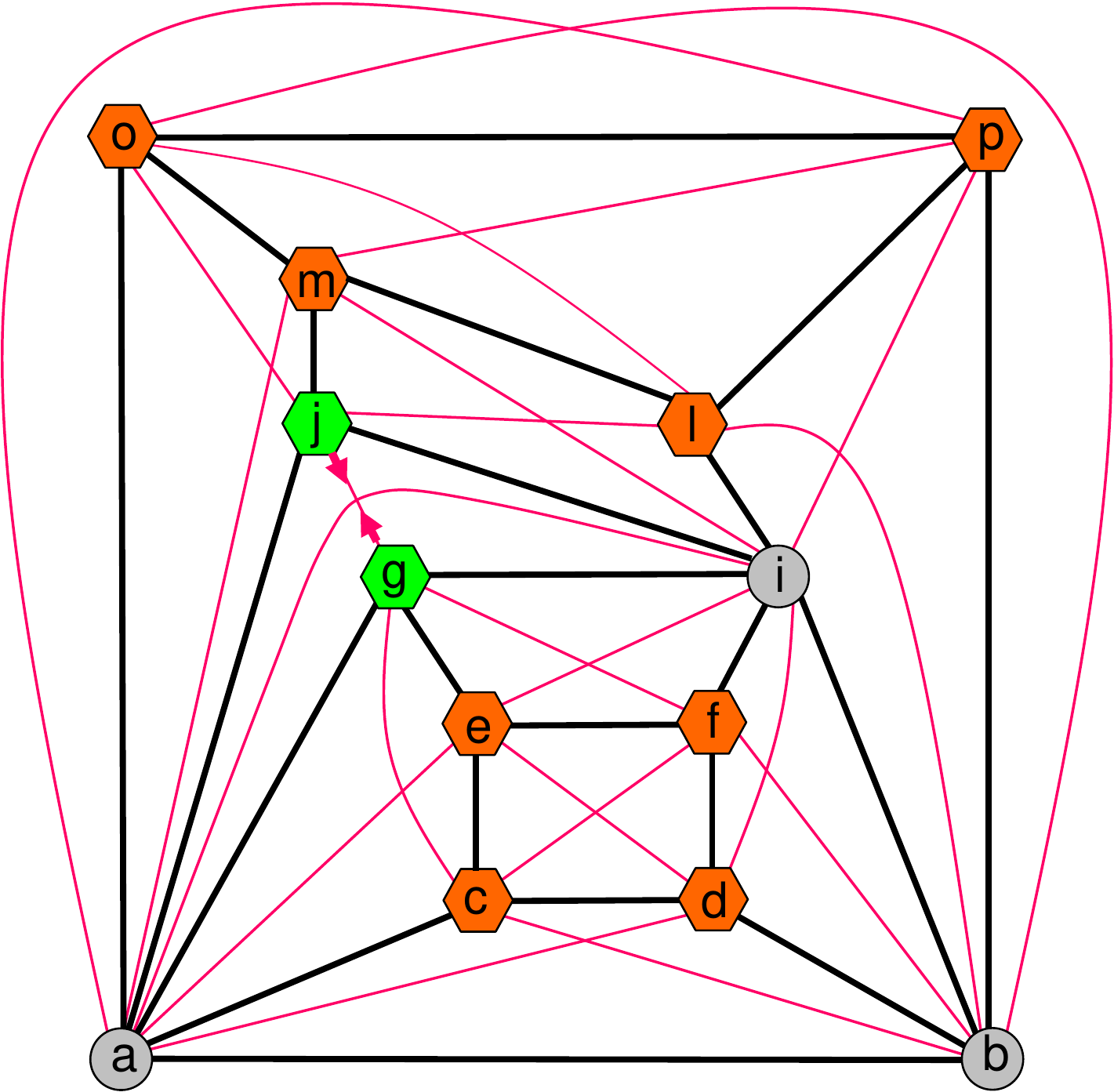}
    }
    \label{fig:G16hi-nl-kj}
  }
\subfigure[and after $SR(j \mapsto g)$. ]{
    \parbox[b]{5.5cm}{%
      \centering
      \includegraphics[scale=0.28]{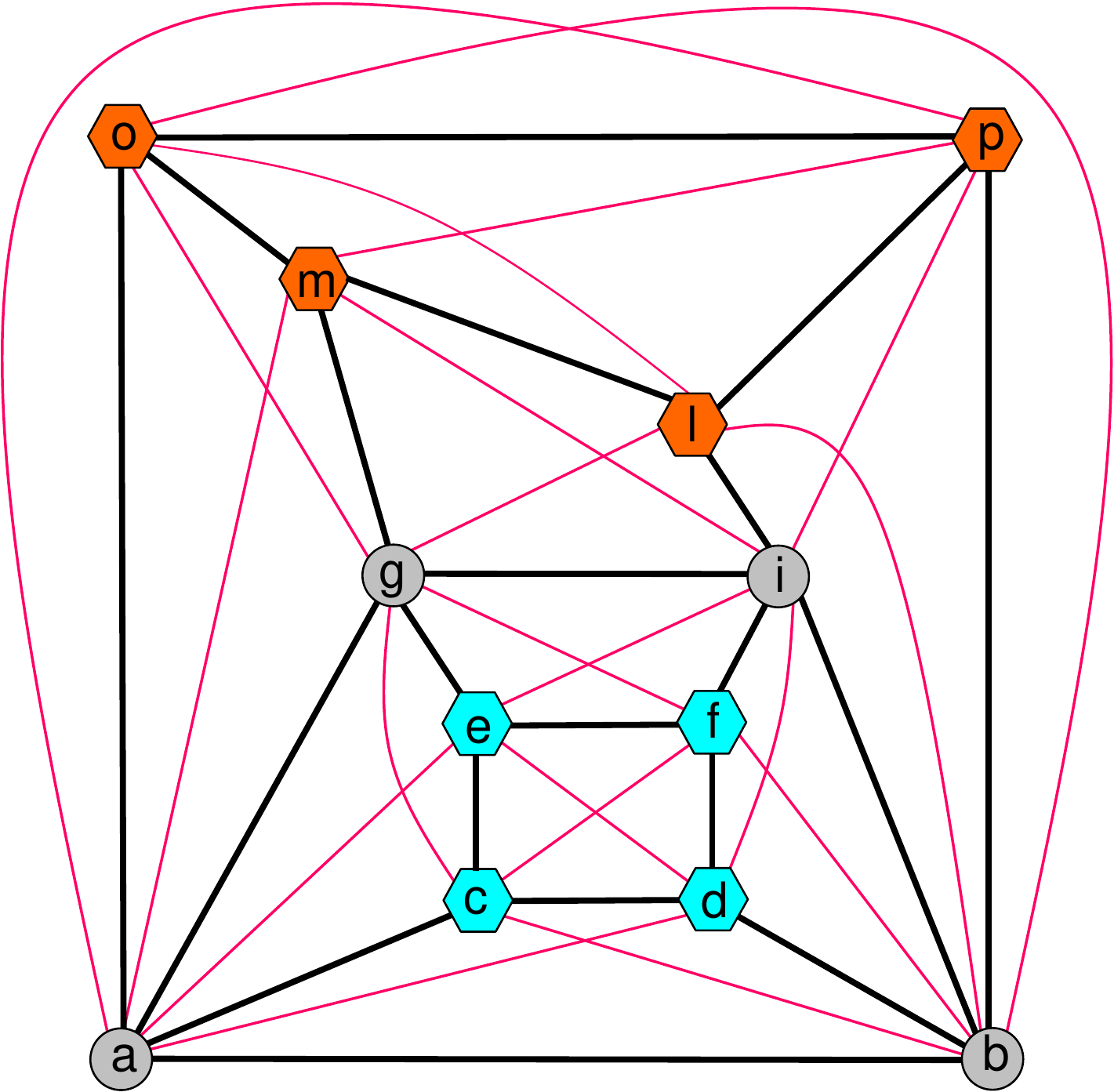}
    }
    \label{fig:G16hi-nl-kj-jg}
  }
  \subfigure[and $XW_6$ after $CR(c,d,e,f)$. ]{
    \parbox[b]{5.5cm}{%
      \centering
      \includegraphics[scale=0.28]{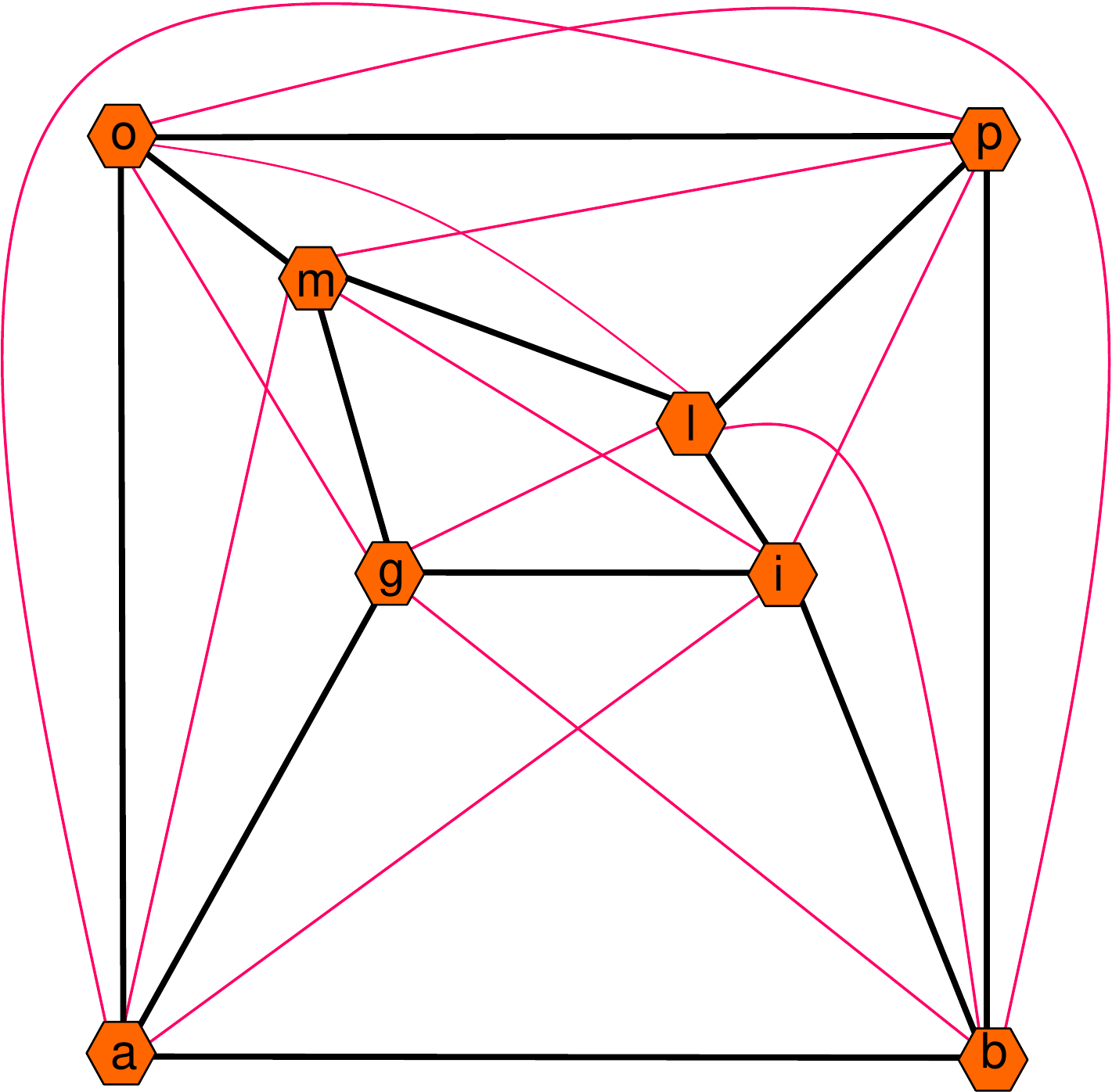}
    }
    \label{fig:G16hi-nl-XW6}
  }
  \caption{A reduction of an input graph to $XW_6$.}
  \label{fig:reduce1-G16}
\end{figure}

Example \ref{ex:example1} shows that the reductions have side
effects.
 A reduction  removes some candidate(s) and may introduce   new ones.  If
an optimal 1-planar graph has several good candidates, then the use
of a reduction may turn some good candidates to bad ones, and vice
versa. In  Example \ref{ex:example1}, this effect is due to the
alternation of the degree of some vertices. If $SR(x \mapsto x_4)$
is applied  as in Fig.~\ref{SR}, then the degree of $x_4$ increases
by two and may change $x_4$ from a candidate to a non-candidate. Bad
candidates in the neighborhood of $x_4$ may turn from bad to good
including four candidates for an application of $CR$. The degree of
$x_3$ and $x_5$ decreases by two and they become a candidate if
their degree changes from eight to six. Simultaneously, good
candidates in their neighborhood from may turn to bad. Similarly, a
$CR$-reduction decreases the degrees of the vertices on the outer
cycle by two which  become a candidate if their degree was eight
before the $CR$-reduction.

Example \ref{ex:example1} also shows that  $SR$-reductions may
destroy 5-connectivity which sheds new light on the results of
Schumacher \cite{s-s1pg-86} as stated in
Thm.~\ref{thm:SRproperties}.

\section{Characterization} \label{charact}

Next, we study combinatorial properties of the graph reduction
systems with the sets of rules  $\{SR, CR\}$ and $\{SR\}$. Unless
otherwise stated, both rules can be used.

A \emph{pre-extended wheel graph} $G$ is an optimal 1-planar graph
such that a single reduction   results in an extended wheel graph.
The respective set is denoted by $preXW$. In particular, a graph is
a $SXW_{2k}$ if $G \rightarrow XW_{2k}$ by an $SR$-reduction and a
$CXW_{2k}$ if $G \rightarrow XW_{2k}$ by a $CR$-reduction. In other
words, $preXW = \cup_{k \geq 4} \{SR^{-1}(XW_{2k}) \, \cup \,
CR^{-1}(XW_{2k}) \} \cup \{CXW_6\}$. Here, $SR^{-1}$ and $CR^{-1}$
are the inverse of $SR$ and $CR$, respectively. Pre-extended wheel
graphs $SXW_{2k}$  were used by Schumacher \cite{s-s1pg-86} for his
reductions of 5-connected optimal 1-planar graphs to $XW_8$. Graph
$CXW_{6}$ is shown in Fig.~\ref{fig:G16hi-nl-kj-jg}.

Note that the minimum extended wheel graph $XW_6$ cannot be obtained
from a pre-extended wheel graph by an $SR$-reduction, since all
vertices of $XW_6$ have degree six and an $SR$-reduction introduces
a vertex of degree at least eight. In addition,  such a graph would
have nine vertices, but there is no optimal 1-planar graph with nine
vertices \cite{bsw-1og-84}. Hence, there is a $CR$-reduction if $ G
\rightarrow XW_6$.

Pre-extended wheel graphs are characterized as follows:

\begin{lemma} \label{lem:charpreXQ}
  A pre-extended wheel graph $CXW_{2k}$ is obtained from $XW_{2k}$ by the
  extraction of a pair of crossing edges in a quadrilateral face with one
  pole and three vertices on the cycle of $XW_{2k}$ and the insertion of
  $CC$. The vertices of the outer cycle of $CC$ and the quadrilateral are
  identified.

  A graph $SXW_{2k}$ is obtained from $XW_{2k}$ by replacing a planar edge $(p,
  v_j)$ and the crossing edges $(p, v_{j-1})$ and $(p, v_{j+1})$ between a
  pole $p$ and three consecutive vertices $v_{j-1}, v_j, v_{j+1}$ on the
  cycle of $XW_{2k}$ by $CS$ and identifying the sequence of vertices
  $(v_j, v_{j+1}, v_{j+2}, p, v_{j-2}, v_{j-1})$ with the vertices on the
  cycle of $CS$.
\end{lemma}

%\noindent \textbf{Proof.}
\begin{proof}
  Suppose that $G \rightarrow XW_{2k}$ by a $CR$-reduction of an optimal
  1-planar graph $G$. Then there are four good candidates $x_1, x_2, x_3,
  x_4$ which together with their neighbors form $CC$. The outer cycle of
  $CC$ is a separating 4-cycle, which can only be built from three
  consecutive vertices of the cycle and a pole of the extended wheel graph.
  Thus, none of the good candidates can be a vertex from the cycle of
  $XW_{2p}$ for some $p \geq k$. Hence, there is a planar quadrilateral into
  which $x_1, x_2, x_3, x_4$ are inserted to form $CC$. Clearly, the
  $CR$-reduction $G \rightarrow XW_{2k}$ is feasible.

  If $SR$ is applied to yield $XW_{2k}$, then  $SR(x \mapsto p)$
  must be applied at some candidate $x$ and a pole $p$,  since only
  a  $SR$-reduction increases the
  degree of a vertex and this were illegal for any vertex on the cycle of
  $XW_{2k}$. Then there is a planar hexagon $(p ,v_{j-2}, v_{j-1}, v_j,
  v_{j+1}, v_{j+2})$ in $G$ into which a new vertex $x$ is inserted to form
  $SC$. Clearly, the $SR$-reduction $G \rightarrow XW_{2k}$ is feasible.
%  \qed
 \end{proof}
%  \hspace{5mm} $\square$ \\

A pre-extended wheel graph allows for different reductions: in a
single step to an irreducible extended wheel graph and with two
$SR$-reductions to the next smaller pre-extended wheel graph. We
adopt the terms pole and vertices on a cycle from extended wheel
graphs and the construction given in Lemma \ref{lem:charpreXQ}.

\begin{lemma}\label{lem:pre}
  For every pre-extended wheel graph $G$, if $G = SXW_{2k}$ with $k \geq 5$,
  then $G \rightarrow XW_{2k}$ and $G \rightarrow^* SXW_{2k-2}$. If $G =
  CXW_{2k}$ with $k \geq 4$, then $G \rightarrow XW_{2k}$ and $G
  \rightarrow^* CXW_{2k-2}$.
\end{lemma}

%\noindent \textbf{Proof.}
\begin{proof}
  Clearly, there is an immediate reduction to an extended wheel graph.
 Alternatively, consider  $G = SXW_{2k}$ with a good candidate $x$ with neighbors $(v_1, v_2,
  v_3, v_4, v_5, p)$, where   $p$ is a pole and $v_1, \ldots, v_5$
  are consecutive vertices on the cycle, as described in Lemma \ref{lem:charpreXQ}.
  Then also $v_2$ and $v_4$ are good candidates for a $SR$-reduction, since
  their planar neighbors $v_1, q$ and $v_5, q$, respectively, have degree
  eight, and there is no blocking edge. Apply $SR(v_4 \mapsto v_6)$ with
  $(v_3, x, v_5, v_6, q, v_2)$ as cycle around $v_4$. This reduction
  increases the degree of $v_6$ to eight and introduces $v_5$ as a good
  candidate with neighbors $(x, v_1, p, v_7, v_6, v_3)$ if $k \geq 5$ and
  $p$ still has degree at least eight. Then $SR(v_5 \mapsto v_7)$ is
  feasible, which results in $SXW_{2k-2}$. However, if $k=4$, then $v_1$ and
  $v_6$ are the new poles of $XW_8$.

  Similarly, consider $CXW_{2k}$ where the inner cycle of $CC$  is inserted into
  a quadrangle $(p, v_4, v_5, v_6)$ with a pole $p$ and three
consecutive
  vertices $v_4, v_5, v_6$ on the cycle of $CXW_{2k}$, see Fig.
  \ref{DW-reductions}. In other words, $CXW_{2k}$ is obtained from
  $XW_{2k}$ by $CR^{-1}$  using $(p, v_4, v_5, v_6)$ as outer
  cycle of $CC$.
  Then $v_4, v_5, v_6$ no longer are candidates,
  whereas the newly inserted vertices $x_1, x_2, x_3, x_4$ are. Each of them
  is good for $CR$ and is blocked for $SR$.

  However, $v_3$ is a good candidate for $SR(v_3 \mapsto v_5)$ and,
  thereafter, $v_4$ is a good candidate for $SR(v_4 \mapsto v_2)$.
  The use of these reductions removes
  $v_3$ and $v_4$. Thereafter, using $CR(x_1, x_2, x_3, x_4)$ results in $XW_{2k-2}$.
%   \qed
 \end{proof}
 % \hspace{5mm} $\square$ \\

\begin{figure}
  \centering
  \subfigure[A $CXW_{2k}$ with good candidates $x_1, x_2, x_3, x_4$ for
    $CR$ and $v_3, v_7$ for $SR$. Good candidates are drawn as a square.]{
    \parbox[b]{5.5cm}{%
      \centering
      \includegraphics[scale=0.21]{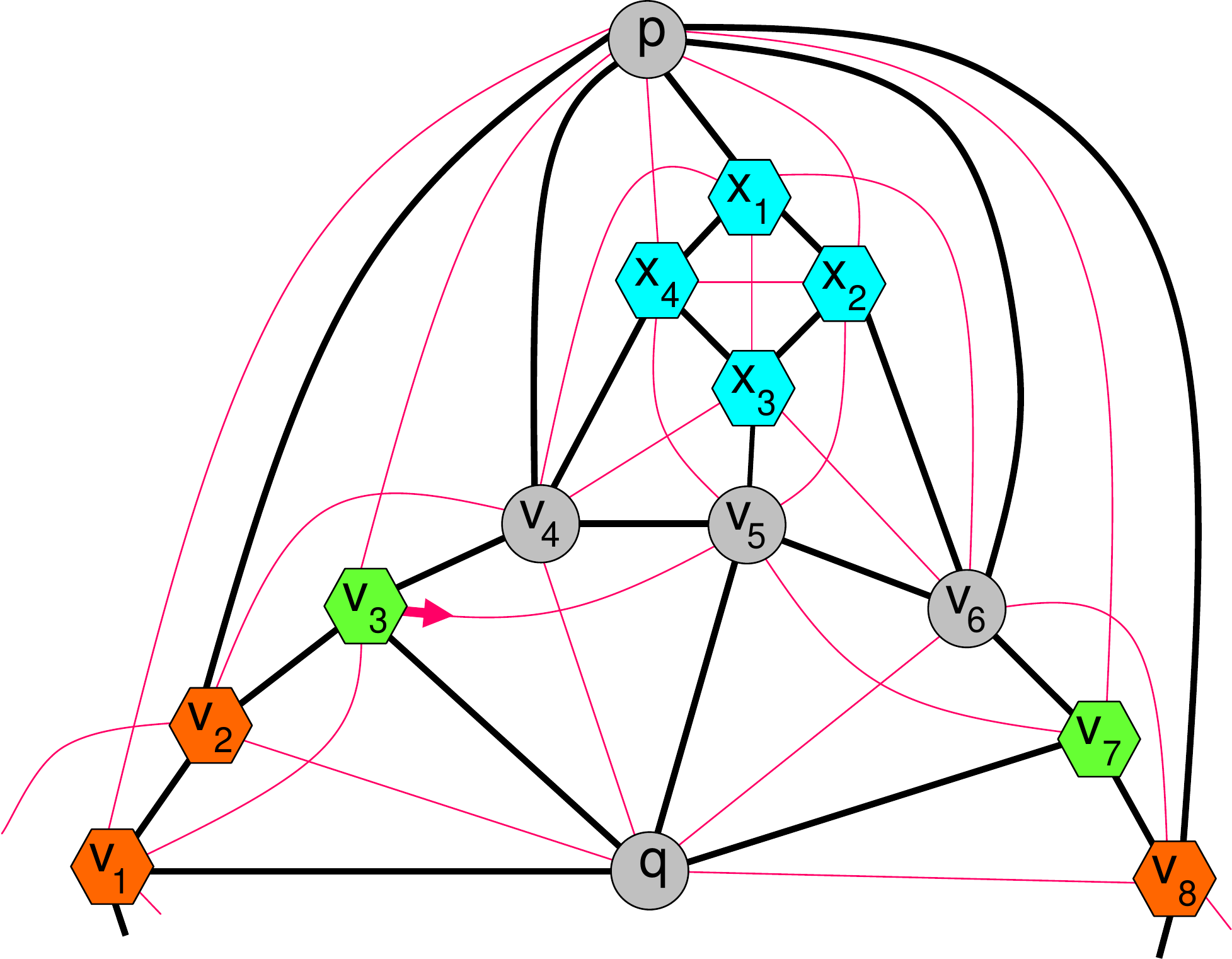}
    }
    \label{DW-reduction-a}
  }
  \hfil
  \subfigure[The graph after $SR(v_3
      \mapsto v_5)$. If  $k > 4$  and $b$ has more neighbors, then $v_7$ is
    good.]{
    \parbox[b]{5.5cm}{%
      \centering
      \includegraphics[scale=0.21]{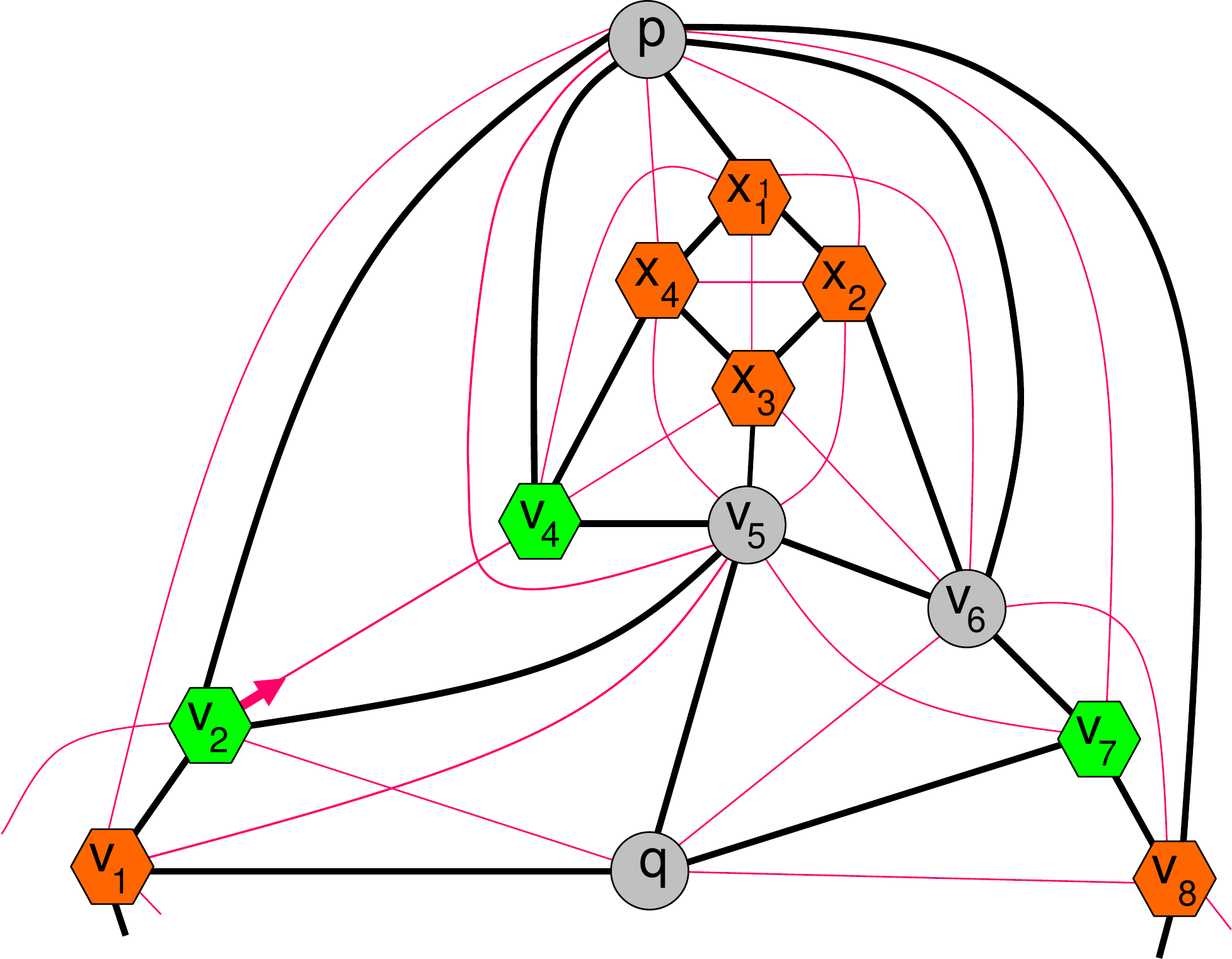}
    }
    \label{DW-reduction-b}
  } \\
  \subfigure[$CXW_{2k-2}$ after $SR(v_2 \mapsto v_4)$. ]{
    \parbox[b]{5.5cm}{%
      \centering
        \includegraphics[scale=0.21]{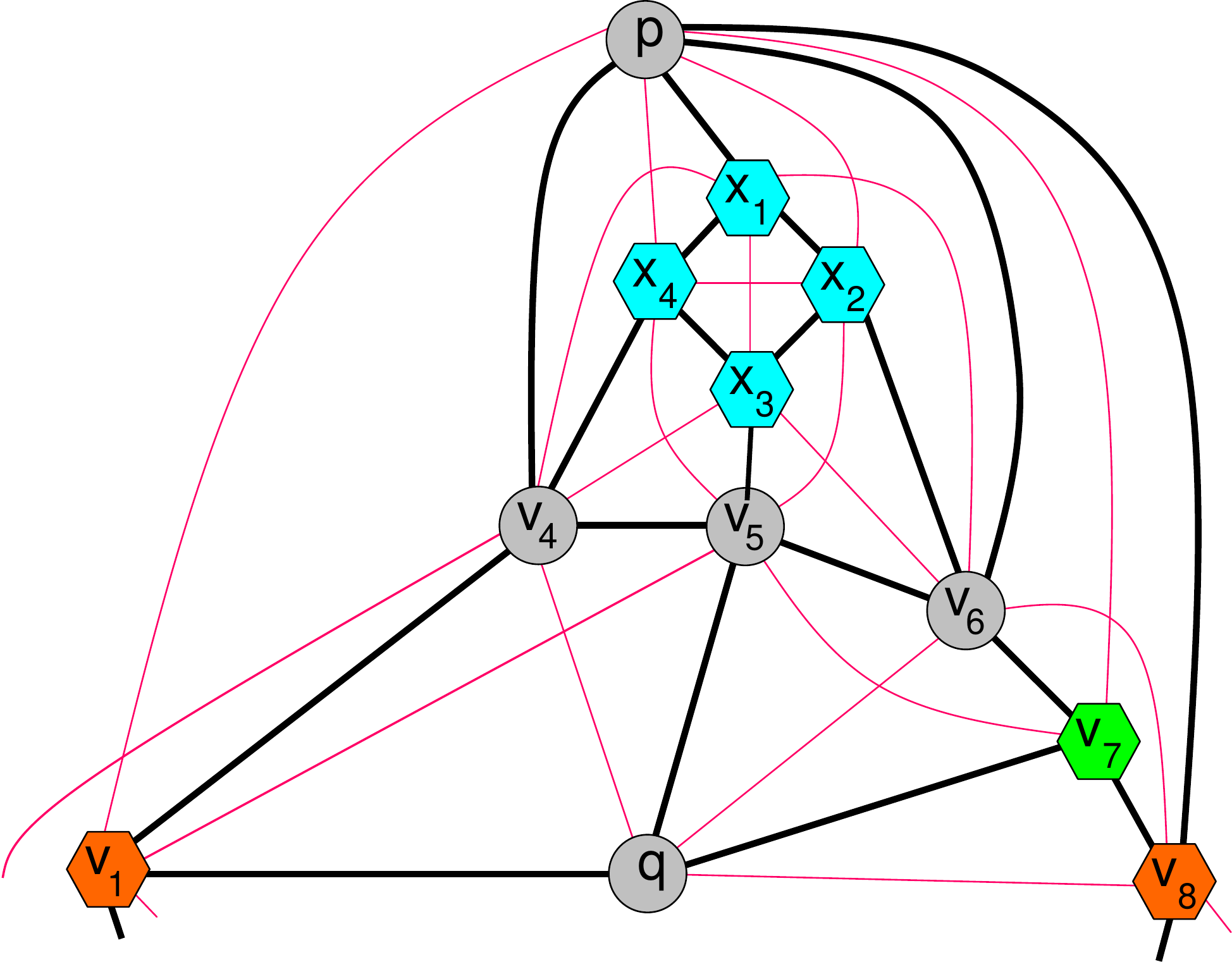}
    }
    \label{DW-reduction-c}
  }
  \hfil
  \subfigure[$XW_{2k-2}$ after the CR-reduction.]{
    \parbox[b]{5.5cm}{%
      \centering
      \includegraphics[scale=0.21]{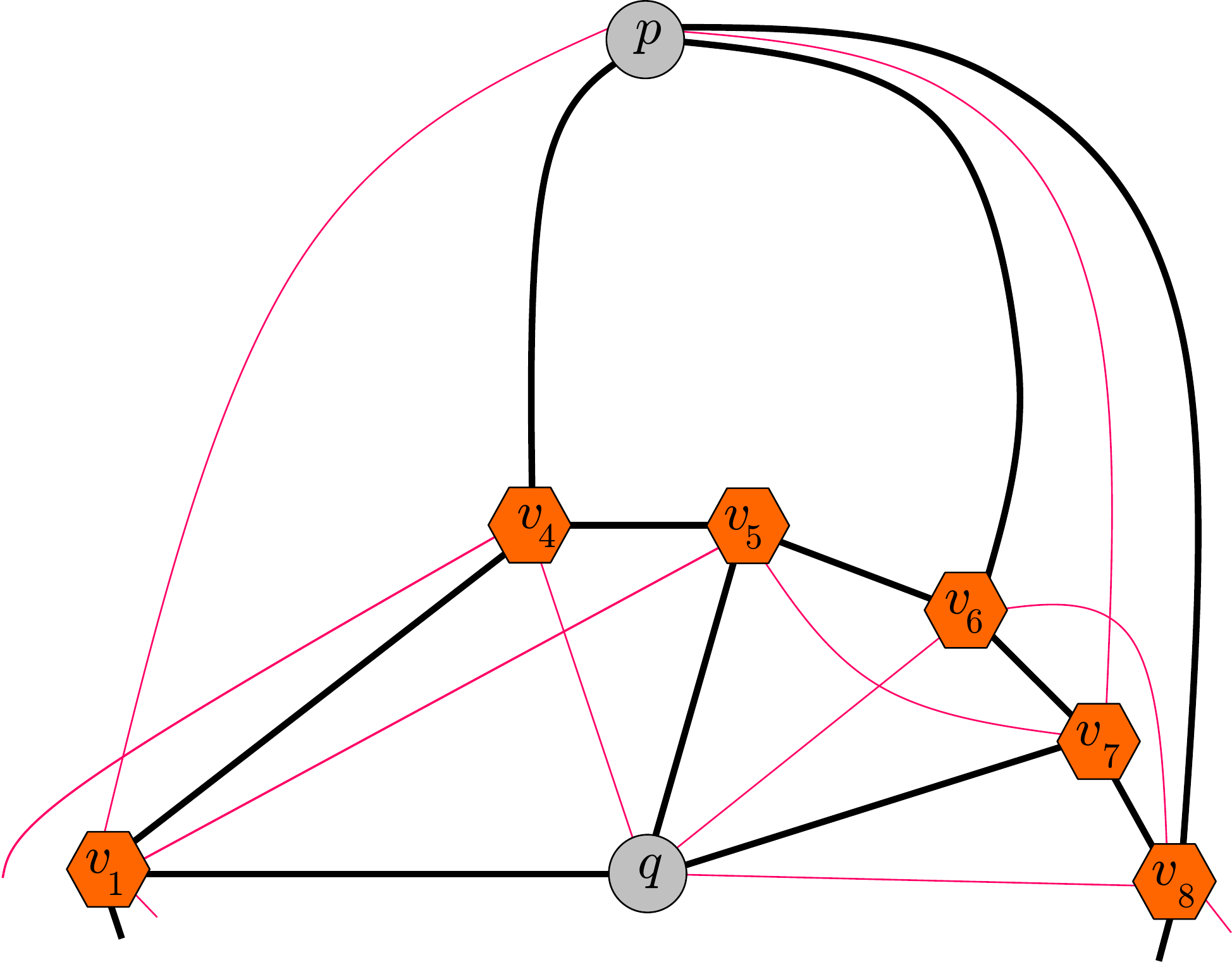}
    }
    \label{DW-reduction-d}
  }
  \caption{Reducing $CXW_{2k}$ to $XW_{2k-2}$ for $k \geq 4$ with poles $p$
    and $q$. The dangling edges at $v_1, v_2$ and $v_8$ are incident to the
    next vertices on the cycle.}
  \label{DW-reductions}
\end{figure}

The reductions $SR(v_3 \mapsto v_5)$ and $SR(v_4 \mapsto v_2)$
followed by a $CR$-reduction  are shown in Fig.~\ref{DW-reductions}.
By symmetry, $SR(v_7 \mapsto v_5)$ and $SR(v_6 \mapsto v_8)$ could
be applied. If $k \geq 5$, then both reductions can be applied
(sequentially or in parallel), which demonstrates the
non-determinism  of the reduction system.\\

Confluence is an important property of rewriting systems, and it is
independent of the objects, i.e., whether they are strings, terms,
polynomials \cite{bn-tr-98, bo-srs-93, bw-gb-93}, or graphs
\cite{e-cfgg-97}.   A rewriting system is \emph{confluent} if $x
\rightarrow^* u$ and  $x \rightarrow^* v$ implies that there is a
common descendant $z$ with $u \rightarrow^* z$ and $y \rightarrow^*
z$. In consequence, if two rules  can be applied at different places
of $x$  starting two
 reductions, then the reductions join at a common descendant. In
particular, if $x$ reduces to an irreducible element $y$, then $y$
is unique.

Example \ref{ex:example1} and Lemma \ref{lem:pre} show  that
reductions to different extended wheel graphs are possible, and
extended wheel graphs are irreducible.

\begin{corollary}
  The reduction system   with  $SR, CR$ ($SR$) is non-confluent on
  (5-connected) optimal 1-planar graphs.
\end{corollary}

In Example \ref{ex:example1}, we have shown that the given graph
$G_{16}$ can be reduced to   $XW_8$ and $XW_6$. Next, we show that
every reducible optimal 1-planar graph $G$ can be reduced to every
extended wheel graph in a range from $XW_{2s}$ to  $XW_{2t}$ where
$s=3$ if $G$ is not $5$-connected and $s=3$ or $s=4$ if $G$ is
$5$-connected.

\begin{theorem} \label{upperbound}
For every reducible optimal
  1-planar graph $G$ there is some upper bound $t$ with
  $t \leq (2n+p+q-4)/8$, where $n$ is the
  size of $G$ and $p$ and $q$ are the two largest degrees of the vertices of
  $G$ and $s \in \{3,4\}$, so that for all $i$ and $s \leq i \leq t$ there
  is a reduction $G \rightarrow^* XW_{2i}$.
\end{theorem}

%\noindent \textbf{Proof.}
\begin{proof}
  There is a reduction $G \rightarrow^* P \rightarrow XW_{2u}$ for some $u$,
  where $P$ is a pre-extended wheel graph. If $u > 4$, then $P \rightarrow
  P'$ for a pre-extended wheel graph $P'$ with $|P'| = |P|-2$ by Lemma
 \ref{lem:pre}. By induction, there is a reduction to a pre-extended wheel
  graph of any smaller size, which reduces to
  $XW_8$ and $XW_6$, respectively, if the pre-extended wheel graph of size
  $12$ is $SXW_8$ and $CXW_6$, respectively. Therefore, we have $s \in
  \{3,4\}$.

  The upper bound on $u$ is due to the fact that the poles of an extended
  wheel graph $XW_{2k}$ have degree $2k$ and the degree of a vertex can only be
  increased by two by    $SR$, which, however, removes one vertex.
  Hence, it takes at least
  $(2k-p)/2 + (2k-q)/2$ reductions to increase the degree of the two vertices
  with the highest degree to two poles of degree $2k$.
  Then, at most  $n - (2k-p)/2 + (2k-q)/2$ vertices remain, which is $2k+2$ for the resulting $XW_{2k}$.
  Hence,
  it takes at least $j = (2n-p-q-4)/4$ $SR$-reductions to transform the two
  highest degree vertices of $G$ into poles of degree $2k$, which results in an $XW_{2k}$
  with $k \leq (2n+p+q-4)/8$.
%   \qed
 \end{proof}
%  \hspace{5mm} $\square$ \\

\begin{corollary}\label{cor:XW6-XW8}
  Every reducible optimal 1-planar graph can be reduced to $XW_6$ or
  $XW_8$.
\end{corollary}

\begin{corollary}\label{cor:XW6-XW8-2}
  Every pre-extended wheel graph $SXW_{2k}$ can be reduced to every
  extended wheel graph  $XW_{2i}$ with $i=4,\ldots, 2k$  using only $SR$-reductions.
 Every pre-extended wheel graph $CXW_{2k}$ can be reduced to every
  extended wheel graph  $XW_{2i}$ with $i=3,\ldots, 2k$.
\end{corollary}

Schumacher \cite{s-s1pg-86} proved that for every $5$-connected
optimal 1-planar graph $G$ there exists a reduction to an extended
wheel graph, even to $XW_8$, using only $SR$-reductions. Corollary
\ref{cor:XW6-XW8} extend this result to all optimal 1-planar graphs.
There exists a good candidate for   $SR$  in the interior of $G$ if
$G$ has no separating $4$-cycle and $G$ is reducible, as proved in
Lemma $4$ of
  \cite{bggmtw-gsqs-05}. In consequence, if  an optimal 1-planar
graph $G$ has separating $4$-cycles and partitions into $G_{in}$ and
$G_{out}$, such that $C$ is an innermost (or outermost) separating
$4$-cycle, then $G_{in}+C$ is $5$-connected and there exists a good
candidate for an $SR$-reduction in $G_{in}$. The candidate is not on
the $4$-cycle $C$. Recall that a completion of $G_{in} + C$ and
$G_{out}+C$ adds the diagonals of the separating $4$-cycle and we
obtain $G_{in}+C = XW_6$ if $G_{in}$ has only four vertices.

On the other hand, the four candidates on the inner cycle of $CC$
mutually block each other for a $SR$-reduction. There are no means
to raise the blockade by an $SR$-reduction. Only $CR$ can do. As an
extension thereof, if $C$ is an innermost (outermost) separating
$4$-cycle of an optimal 1-planar graph $G$, then the vertices of $C$
are ``frozen'' for   $SR$-reductions and remain if $G_{in}+C$ is
reduced to an extended wheel graph, say $XW_8$. However, further
reductions are possible after a recombination of the results. This
leads to the following facts:

\begin{theorem} \label{thm:charact-4-connected}
If a reducible optimal 1-planar graph  has a separating $4$-cycle,
then there is a reduction  to $XW_6$.
\end{theorem}

%\noindent \textbf{Proof.}
\begin{proof}
Let $C = (v_1, v_2, v_3, v_4)$ be an innermost separating $4$-cycle
such that $G-C$ is partitioned into $G_{in}$ and $G_{out}$. Then
there is a good candidate $x$ for $SR$ in $G_{in}+C$  and $x$ is not
on $C$, as proved in \cite{bggmtw-gsqs-05}, Lemma 4. Hence, there is
an $SR$-reduction of $G_{in}+C$ to $XW_8$ preserving $C$. Exactly
one pole of $XW_8$ is in $C$. If all vertices of $C$ were
candidates, then  $XW_8$ cannot be realized. Both poles cannot be in
$C$, since the poles are not connected by an edge whereas the
subgraph induced by the vertices of $C$ is $K_4$, since a pair of
crossing edges is added in the outer face.

Consider graph $H$ obtained from $G$ by replacing $G_{in}$ by
$XW_8$. Then $G \rightarrow^* H$ where the $SR$-reduction steps of
$G_{in}+C \rightarrow^* XW_8$ are applied. Now, we can remove
$G_{in}$ from $H$ by two $SR$-reductions and a final $CR$-reduction.
The vertices on $C$ have degree at least eight, since they have two
black neighbors on $C$ and at least one black neighbor in $G_{out}$
and $XW_8$ obtained from $G_{in}$, and there exists a candidate for
$SR$ in the interior of $C$ according to \cite{bggmtw-gsqs-05}.

So we proceed towards the outermost separating 4-cycle $C_{out}$ and
 reduce the inner and the outer components $H_{in}$ and $H_{out}$
to $XW_8$ preserving $C_{out}$. The poles of $XW_8$ from the inner
and outer components on $C_{out}$ may or may not coincide. Reduce
the inner subgraph of $C_{out}$ by two $SR$-reductions to a single
4-cycle which together with $C_{out}$ forms $CC$, and similarly for
the outer subgraph.  A final $CR$-reduction to the inner 4-cycle
yields $XW_6$.
%  \qed
 \end{proof}
%  \hspace{5mm} $\square$ \\

Next, consider the ``nested extended quadrangles'' in
Fig.~\ref{fig:nested}, which are optimal 1-planar graphs and can be
reduced to $XW_6$ by $CR$-reductions. The vertices on the innermost
and   outermost cycles have degree six and thus are candidates,
whereas the other vertices have degree eight. However, the
candidates have type $4$ and an $SR$-reduction is infeasible whereas
$CR$ can be applied. This property remains if a $CR$-reduction is
used and, thereby, there is a unique reduction  to  $XW_6$. In
consequence, we obtain:

\begin{theorem} \label{reduce5connected}
There is an infinitely many optimal 1-planar graphs $G$ such that a
reduction of $G$ to an extended wheel graph $XW_{2k}$ implies $k=3$
and only $CR$-reductions can be applied.
\end{theorem}

The nested extended quadrangles are inaccessible to $SR$-reductions,
since the candidates pairwise block each other  by red blocking
edges. In Example \ref{ex:example1} we have shown that
$SR$-reductions may introduce a separating 4-cycle and destroy
5-connectivity. For Schumacher's reduction system we obtain:

\begin{theorem} \label{thm:SRproperties}
Suppose that only $SR$-reductions are used.
\begin{enumerate}
\item For every 5-connected optimal 1-planar graph $G$ there exists a
reduction to an extended wheel graph, and even to $XW_8$.
\item  There are 5-connected optimal 1-planar
graphs $G$ that are reduced to  optimal 1-planar graphs with
separating 4-cycles.
 \item If $G$ has a separating 4-cycle and $G
\rightarrow G'$ by  an $SR$-reduction, then $G'$ has a separating
4-cycle.
\end{enumerate}
\end{theorem}

%\noindent \textbf{Proof.}
 \begin{proof}
The first statement has been proved by Schumacher \cite{s-s1pg-86}.
Statement (2) is shown by Example \ref{ex:example1}. Finally,
suppose there is a separating 4-cycle $C = (a,b,c,d)$ in $G$. For a
removal, an $SR$-reduction must be applied to a vertex of $C$.
Suppose that $a$ is a candidate with neighbors $b,v,d, w_1, w_2,
w_3$ and $SR(a \mapsto v)$. Thereafter, $C' = (w_2, b,c,d)$ is a
separating 4-cycle.
%  \qed
\end{proof}
 % \hspace{5mm} $\square$ \\

In consequence, the following properties hold for  Schumacher's
reduction:

\begin{corollary} \label{cor:SRproperties}
(1) The 5-connected optimal 1-planar graphs are not closed under
 $SR$-reductions. \\
\noindent (2) If $G$ is  an optimal 1-planar graph $G$ with a
separating 4-cycle and $H$ is obtained from $G$ by using only
$SR$-reductions, then $H$ is not an extended wheel graph.
\end{corollary}

Hence, the extended wheel graphs are not the set of irreducible
optimal 1-planar graphs under $SR$-reductions. Schumacher's
presupposition \cite{s-s1pg-86} and restriction of $SR$-reductions
to 5-connected optimal 1-planar graphs is necessary. Reductions
towards an extended wheel graph  using only $SR$ get stuck if there
is a 4-cycle.

The extended wheel graphs constitute an infinite set of irreducible
graphs for  $SR$ and $CR$, however, the optimal 1-planar graphs
constitute   a single equivalence class with $XW_6$ as a
representative, even if the equivalence relation is defined only by
$SR$. Let $G_1 \sim  G_2$ if and only if $G_1$ can be transformed
into $G_2$ by a sequence of feasible applications of $SR, \,
SR^{-1}, \, CR$ and $CR^{-1}$, respectively. Here, $G \rightarrow
G'$ by $SR^{-1}$ is feasible if $G' \rightarrow G$ by $SR$ is
feasible, and similarly for $CR^{-1}$. Recall that the inverse
reductions are  the $Q_v$-splitting and the $Q_4$-cycle addition  of
Suzuki \cite{s-rm1pg-10}.

\begin{theorem} \label{thm:transform}
  A graph $G$ is optimal 1-planar if and only if  $G$ is equivalent
  to the minimum extended wheel graph $XW_6$, where the equivalence
  relation is defined by feasible applications of $SR, CR$, and their
  inverse.
\end{theorem}

%\noindent \textbf{Proof.}
\begin{proof}
  First, reduce $G$ to a small $XW_{2k}$ using $SR$ and $CR$-reductions,
  where $k=3,4$. If $k \geq 4$ then expand $XW_{2k}$ to a pre-extended wheel graph
  $CXW_{2k}$, which is then reduced to $XW_6$ according to Corollary \ref{cor:XW6-XW8}.
% \qed
\end{proof}
%  \hspace{5mm} $\square$ \\

Although $SR$-reductions may get stuck on 4-connected graphs, they
can cope with them under equivalence.

\begin{theorem} \label{thm:transform}
  A graph $G$ is optimal 1-planar if and only if  $G= XW_6$ or $G$
  is $SR$-equivalent
  to     $XW_8$, where the $SR$-equivalence
  relation is defined by feasible applications of $SR$ and
  $SR^{-1}$.
\end{theorem}

%\noindent \textbf{Proof.}
\begin{proof}
The ``if'' direction follows from the fact that feasible
applications of $SR$ and
  $SR^{-1}$ preserve optimal 1-planar graphs.
For the ``only if'' direction, first resolve all separating 4-cycles
by an  $SR^{-1}$-reduction as in Example \ref{ex:example1}. For $C =
(a,b,c,d)$  consider the planar quadrangles defined by $a,b, u_1,
u_2$ and $a,b, v_1, v_2$. Then $SR^{-1}$ introduce a center $x$ with
neighbors $a,u_1, u_2, b, v_2, v_1$ in circular order and $SR(x
\mapsto a)$ or $SR(x \mapsto b)$ is feasible. Thereafter, there is a
5-connected optimal 1-planar graph $H$ that can be reduced to $XW_8$
by $SR$-reductions.
%
% \qed
 \end{proof}
%    \hspace{5mm} $\square$ \\

\begin{corollary}
The equivalence problem $G_1  \sim  G_2$  can be solved in linear
time. In addition, if $G_1  \sim  G_2$ then a transformation of
$G_1$ into $G_2$ can be computed in linear time.
\end{corollary}

%\noindent \textbf{Proof.}
 \begin{proof}
Two graphs are equivalent if and only if both are optimal 1-planar
which can be solved in linear time \cite{b-ro1plt-16}. For the
transformation, an extended wheel graph is first transformed into a
pre-extended wheel graph $CXW_{2k}$ using $CR^{-1}$. Thereafter, we
use Lemma \ref{lem:pre} and transform both graphs into $XW_8$ and
then concatenate the transformations taking the inverse of the rules
from the reduction of $G_2$.
%
%  \qed
\end{proof}
 % \hspace{5mm} $\square$ \\

\section{Conclusion and Perspectives} \label{sect:conclusion}

We have shown that the required feasible use of
 $SR$ and $CR$ reductions can be expressed by local conditions on
 the context of the removed vertices. The reduction system with the rules $SR$ and $CR$
 is context-sensitive, non-deterministic and non-confluent, but, nevertheless, reductions
 can be computed in linear time  \cite{b-ro1plt-16}. Moreover,
  every reducible optimal 1-planar graph can be reduced to any
irreducible extended wheel graph $XW_{2k}$ in a range from $s \leq k
\leq t$, where $s=3$ or $s=4$ and $t$ depends on the given graph.
 Our results generalize
similar ones of Schumacher \cite{s-s1pg-86} who used only
$SR$-reductions that are restricted to 5-connected optimal 1-planar
graphs, but $SR$ reductions do preserve this class.

It would be interesting to see whether similar results hold for
other classes of optimal graphs such as optimal $k$-planar graphs,
which allow up to $k$ crossings per edge \cite{pt-gdfce-97}, or
optimal $IC$ planar graphs, which are the restriction of optimal
1-planar graphs to  independent crossings  where each vertex is
incident to at most one crossing edge \cite{ks-cpgIC-10,
bdeklm-IC-16} and have the maximum of $13/4\,n-6$ edges.

%% --------------------------------------------------------------------
%       Bibliography
%% --------------------------------------------------------------------

\bibliographystyle{splncs03}
\bibliography{brandybibLNCSV5}

\end{document}